\documentclass[11pt,reqno]{amsart}
\pdfoutput=1

\usepackage{caption}
\usepackage{subcaption}
\usepackage[round]{natbib}

\usepackage{amsmath,amsfonts,amssymb,commath}
\usepackage{graphicx}
\usepackage{booktabs}
\usepackage{multirow}
\usepackage{pdfpages}

\makeatletter
\def\bbordermatrix#1{\begingroup \m@th
  \global\let\perhaps@scriptstyle\scriptstyle
  \@tempdima 4.75\p@
  \setbox\z@\vbox{%
    \def\cr{%
      \crcr
      \noalign{%
        \kern2\p@
        \global\let\cr\endline
        \global\let\perhaps@scriptstyle\relax
      }%
    }%
    \ialign{$\make@scriptstyle{##}$\hfil\kern2\p@\kern\@tempdima
      &\thinspace\hfil$\perhaps@scriptstyle##$\hfil
      &&\quad\hfil$\perhaps@scriptstyle##$\hfil\crcr
      \omit\strut\hfil\crcr
      \noalign{\kern-\baselineskip}%
      #1\crcr\omit\strut\cr}}%
  \setbox\tw@\vbox{\unvcopy\z@\global\setbox\@ne\lastbox}%
  \setbox\tw@\hbox{\unhbox\@ne\unskip\global\setbox\@ne\lastbox}%
  \setbox\tw@\hbox{$\kern\wd\@ne\kern-\@tempdima\left[\kern-\wd\@ne
    \global\setbox\@ne\vbox{\box\@ne\kern2\p@}%
    \vcenter{\kern-\ht\@ne\unvbox\z@\kern-\baselineskip}\,\right]$}%
  \null\;\vbox{\kern\ht\@ne\box\tw@}\endgroup}
\def\make@scriptstyle#1{\vcenter{\hbox{$\scriptstyle#1$}}}
\makeatother

\usepackage[margin=1in]{geometry}

\newtheorem{theorem}{Theorem}[section]

\newtheorem{proposition}[theorem]{Proposition}

\newtheorem{definition}[theorem]{Definition}
\newtheorem{lemma}[theorem]{Lemma}

\newtheorem{assumption}[theorem]{Assumption}
\def\pilbar{{\underline\pi}}  \def\piubar{{\overline\pi}}

\def\sigmaubar{{\overline\sigma}}

\def\Pibar{{\overline\Pi}}  \def\alphalbar{{\underline\alpha}}
\def\betalbar{{\underline\beta}}
    \def\epsilonbar{{\overline\epsilon}}
   \def\bigO{{\sf O}}

\def\alphalbar{{\underline\alpha}}  \def\alphaubar{{\overline\alpha}}
\def\pbar{{\overline p}}
   
\def\Expect{{\sf E}}
\def\Prob{{\sf P}}
\long\def\notes#1{\ifinner
           {\tiny #1}
           \else
           \marginpar{\parbox[t]{\noteWidth}{\raggedright\tiny #1}}
       \fi\typeout{#1}}
  
\def\betalbar{{\underline\beta}}

\def\epsilonhat{{\hat\epsilon}}

\def\epsilonhat{{\hat\epsilon}}  \def\epsilonbar{{\overline\epsilon}}
   
\def\Expect{{\sf E}}
\def\Prob{{\sf P}}
\long\def\notes#1{\ifinner
           {\tiny #1}
           \else
           \marginpar{\parbox[t]{\noteWidth}{\raggedright\tiny #1}}
       \fi\typeout{#1}}
  
\def\betalbar{{\underline\beta}}

\def\epsilonhat{{\hat\epsilon}}

\begin{document}

\title[Collusive Outcomes Without Collusion]{Collusive Outcomes Without Collusion: \\ Algorithmic Pricing in a Duopoly Model}

\author{In-Koo Cho}
\address{ Department of Economics \\ Emory University}
\author{ Noah Williams}
\address{Miami Herbert Business School \\University of Miami }

\date{\today}

\begin{abstract}
We develop a model of algorithmic pricing which shuts down every channel for explicit or implicit collusion, and yet still generates collusive outcomes. We analyze the dynamics of a duopoly market where both firms use pricing algorithms consisting of a parameterized family of model specifications. The firms update both the parameters and the weights on models to adapt endogenously to market outcomes. We show that the market experiences recurrent episodes where both firms set prices at collusive levels. We analytically characterize the dynamics of the model, using large deviation theory to explain the recurrent episodes of collusive outcomes. Our results show that collusive outcomes may be a recurrent feature of algorithmic environments with complementarities and endogenous adaptation, providing a challenge for competition policy.
\end{abstract}

\maketitle

\section{Introduction}
In early February 2024, Senator Amy Klobuchar (D-MN), Chairwoman of the Senate Judiciary Subcommittee on Competition Policy, Antitrust, and Consumer Rights, introduced the Preventing Algorithmic Collusion Act.
\begin{quotation}
  ``Price fixing is illegal under our antitrust laws, but the development of automated price-setting algorithms can create loopholes in current law that could be used to unfairly raise prices on everything from rent to rideshares,'' said Klobuchar. ``My bill will strengthen antitrust law and guarantee needed transparency to prevent companies from using algorithms to fix prices to ensure consumers are able to get the full benefits of competition.'' \citep{Klobuchar2024}
\end{quotation}
The proposed legislation focuses on algorithms which use non-public competitor data, which could potentially be used by pricing algorithms to engage in a form of tacit collusion, without needing an explicit agreement or communication. While there seems to be agreement that the current legal and regulatory framework in the United States cannot address algorithmic collusion which arises in the absence of communication, this particular proposed legislation is just one example of the increasing interest in adapting the current framework.\footnote{See \citet*{Harrington2018}, \citet*{EzrachiStucke2020}, \citet*{MackayWeinstein2022}, and \citet*{Stucke2023} for some discussions. For example, \citet*{Stucke2023} states ``...[P]ricing algorithms continuously monitor and respond to competitors' online pricing, thereby inhibiting rivals from gaining sales through discounts. Detecting and challenging such collusion is challenging, as conscious parallelism alone does not violate the Sherman Act. There must be proof of an agreement.''} Similar legal and regulatory changes to address algorithmic collusion have also been discussed in the EU and the OECD.\footnote{See \citet{Calzolari2021} and \citet{OECD2017}, for example.}

The legal and regulatory interest in pricing algorithms draws on both empirical and theoretical work in economics showing that algorithmic pricing may lead to firms to set supra-competitive prices.\footnote{See \citet*{Deng2024} for a good recent survey. We discuss the related literature more below.}  While some types of algorithms and some forms of algorithmic interactions may be amenable to legal or regulatory sanction, in this paper we illustrate the limits of competition policy. We develop an example that rigorously eliminates all possible channels of collusion between firms, but yet
generates collusive outcomes. Coordinated price increases to collusive levels in markets populated by algorithmic price setters therefore cannot be used as evidence of collusion, unless augmented by other evidence such as the inside information about the algorithm itself.

Any legal or regulatory framework must distinguish between outcomes and the mechanisms that lead to them. We define collusive outcomes as prices above the competitive level, which are also commonly known as supra-competitive outcomes. We use the former term because the collusive outcome has precise meaning in our model, as the cartel or joint monopoly outcome. As for the mechanisms that generate outcomes, \citet{Harrington2018} provides a useful definition: ``Collusion is when firms use strategies that embody a reward–punishment
scheme which rewards a firm for abiding by the supra-competitive outcome and punishes it for departing from it.''

In order to achieve implicit or explicit collusion, we typically need one or more of the following conditions to hold:
\begin{itemize}
\item Communication between firms. Even though explicit communication is blocked by law, the firms can implicitly communicate with each other via rational expectations about their opponent's strategy.
\item Agreement. Firms must agree over the target outcome and the punishment needed to enforce the agreement.   Agreement can be explicit or tacit.
\item Patient players.   Unless a firm cares about future profits, no punishment is effective.  Without a credible punishment, no collusive agreement can be sustained.
\item Ability to handle complex agreements.  The punishment requires changing actions according to the history of outcomes.  An oligopolistic firm may have to predict the response of the competing firm in response to violation of the agreement, conditioned possibly on a long history of actions to implement a credible punishment.
\end{itemize}

We construct a model of algorithmic interaction that shuts down every one of these possible channels of collusion, yet recurrently generates collusive outcomes. We show that two myopic firms that independently choose an algorithmic pricing specification that they update over time can be led to jointly increase prices to collusive levels. As we discuss, our model relies on the combination of endogenous algorithmic selection and adaptation with an endogenous data generating process.

Our paper is related to the previous work in economics analyzing algorithmic collusion. Some of that work has shown that algorithmic pricing entails an aspect of commitment, where prices are set more frequently than algorithms are revised. This allows for algorithms to learn opponents' strategies, as in  \citet*{Salcedo2015} and \citet*{LambaZhuk2023}, or for asymmetries in frequency or commitment to lead to collusive outcomes as in \citet*{BrownandMacKay2023}.\footnote{\citet{Leisten2024} studies the interaction between algorithmic pricing and human overriding, showing that it can lead to Edgeworth price cycles similar to the outcomes in our model.} Our model lacks such commitment: firms do not learn each other's strategies, but instead treat the price process of the competing firm as an exogenous process.

More closely related to our paper is work that has illustrated through simulations how collusive outcomes can arise, seemingly spontaneously, through the interaction of firms using algorithms that adapt prices over time. Some important papers here include \citet*{WaltmanandKaymak2008}, \citet*{HansenMisraandPai2021}, \citet*{CalvanoetalAER2020}, and \citet*{AskerFershtmanandPakes2023}.\footnote{While the other papers consider versions of profit maximization, \citet*{HansenMisraandPai2021} assumes regret minimization behavior, which is identical to the aspiration learning behavior.   \citet*{ChoandMatsui05} examined aspiration learning dynamics in the prisoner's dilemma game to show that the collusive outcome is the unique stable solution under a general condition.}
Although they differ in details, these papers share certain features:
\begin{enumerate}
\item They consider duopoly markets with strategic complements.
The best response of one firm is an increasing function of the
action of the other firm.\footnote{\citet*{BanchioandSkrzypacz2022} consider a related model of algorithmic competition in auctions.}
\item They impose an exogenous algorithmic specification.
The authors make explicit assumptions on how each firm ``models the behavior of the other firm,'' which in turn determines the dynamics.\footnote{A reinforcement learning algorithm (or Q-learning algorithm) is a popular specification. Other learning rules such as multi-armed bandit problems are used as well.}
\item Their analysis is based on simulations. Simulations of the algorithms produce an outcome that deviates from one shot Nash equilibrium with a positive frequency.
\end{enumerate}
Recent work by \citet*{BanchioandMantegassa2023} and \citet*{Possnig2023} provides some analytic results in related settings, which we discuss in more detail below.

We keep feature (1), but change (2) and (3). Rather than maximizing a discounted profit stream, we consider two duopoly firms that are myopic optimizers, which makes collusion even more challenging.  Instead of imposing a particular algorithmic structure, we allow the firms to adapt the structure of their algorithms over time. We want to understand why a firm chooses a particular specification of the opponent’s behavior. Finally, instead of numerical analysis, we pursue an analytic approach to understand precisely how these features generate collusive outcomes between algorithmic duopolists, offering new insight beyond what the numerical analysis in previous work can provide.

We focus on a pricing algorithm, as defined in the Congressional legislation discussed above:
\begin{quotation}
The term ``pricing algorithm'' means any computational process, including a computational process derived from machine learning or other artificial intelligence techniques, that processes data to recommend or set a price or commercial term that is in or affecting interstate or foreign commerce. \citet*{USCongress2024}
\end{quotation}
Instead of the reinforcement learning algorithms studied in the previous literature, we assume that each firm has a family of parametric specifications of its opponent's pricing reaction function.
Our model is built on \citet{Williams00}, who first observed collusive outcomes generated by duopolists who are learning a linear reaction function of the competing firm. But instead of imposing a particular model, we assume here that a duopolist entertains two possible specifications: a ``Nash reaction'' where the competing firm does not respond to its price, and a ``linear reaction'' where the competing firm's reaction function is a linear function of its own price.  Facing model uncertainty, a duopolist averages the two specifications by assigning them a probability weight.
The firm optimizes static profits under each specification, and chooses a weighted best response.  Then, based on observations of realized prices and sales, each firm updates both the parameters of each specification and the probability weights. We study the dynamics of the adjustment process where prices, parameters, and specifications all evolve endogenously over time.

\begin{figure}
\centerline{\includegraphics[width=0.7 \textwidth]{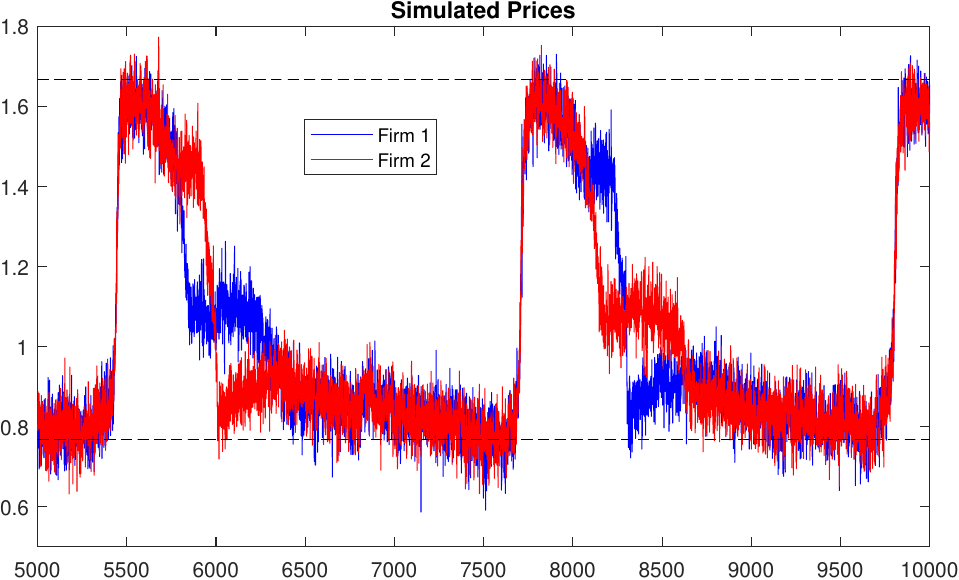}}
\caption{Simulated prices in the duopoly model, with gain $\lambda=0.01$ and shock variance $\sigma^2=0.0025$.}\label{fig-sim-prices}
\end{figure}

Our specification allows us to obtain precise analysis of both the limits and time paths of the evolving duopoly market. The simple nature of the algorithms employed and the market structure of our model make it easy to identify and eliminate any possible channels of collusion.

We start with an observation that the static Nash equilibrium (or competitive outcome) is the unique stable stationary point of the learning dynamics.  Thus, if the economy is initialized at the Nash equilibrium, the prices should stay in a small neighborhood of the Nash equilibrium with probability close to 1 in the limit.  We also show that if the duopolists are committed to the ``Nash specification,'' recurrent escapes to the collusive outcome do not arise.

But the algorithmic interactions lead to outcomes which are very different from these Nash limit predictions.
As in the sample simulation shown in Figure \ref{fig-sim-prices},  the model has ``rare'' but recurrent episodes where prices increase to collusive levels. These outcomes arise in the algorithmic specifications that allow for non-trivial feedback. In these episodes, each duopolist's ``linear reaction'' specification captures that the sample paths of the pair of prices are highly correlated, which leads to the collusive outcome.

Since each duopolist chooses a specification and its price independently, the sample path of the two prices should be independent most of the time.  A collusive outcome is not an equilibrium of the algorithmic dynamics, as it requires highly correlated pricing rules.  Yet, it is a recurrently observed event.  Importantly, the subjective weight assigned to  ``linear reaction'' specifications converges to one, as the algorithm learns from the market data.   Thus recurrent episodes of collusive prices are a robust feature of our model.

Although the collusive outcome is an integrated part of the algorithmic dynamics, such outcomes are not explained by common limit results applied
in the literature. These methods focus on analysis of the associated deterministic dynamics of the stochastic system, such as the fluid equation or the associated ordinary differential equation (ODE).\footnote{See \citet{Benaim96} for a discussion of the deterministic dynamics of stochastic approximation and \citet{Meyn07} for an overview of fluid limits. Related methods have been widely used in economics, for example in \citet{MarcetandSargent89a}, \citet{EvansandHonkapohja00}, and \cite{FudenbergandLevine98}.}  The deterministic component of the algorithmic outcomes converge to the competitive, static Nash prices in the limit, implying that the static Nash prices are the stable stationary point of the learning dynamics.

Collusive outcomes occur for ``non-limit'' sample paths.  Our goal is to understand how the optimal forecasting scheme and the interaction between the data generating process and the algorithm trigger the collusive outcome.   We show that the large deviation properties of the stochastic approximation precisely identify the mechanism of the escape from the competitive outcome to the collusive outcome.\footnote{Our approach explains why the episode of the escape to a collusive outcome vanishes as the number of oligopolistic firms increases.   \citet*{BanchioandMantegassa2023} investigated the dynamics ``away'' from the stable point but fell short of formally explaining the switching mechanism from the stable point to the (rare) escape path.  Large deviation theory explains the ``switching mechanism,'' which is crucial to our analysis.}   Since the data generating process is endogenous, we use stochastic approximation to analyze its asymptotic properties. The collusive outcomes is a ``rare'' event, whose probability vanishes in the limit.   Thus the analysis of the deterministic limit dynamics cannot characterize such events. By examining the large deviation properties of the algorithm as in \citet*{ChoWilliamsandSargent02}, \citet*{ChoandKasa17}, and \citet{Williams2019}, we explain the frequency of collusive outcomes, and why the outcomes differ across model specifications. We also characterize the ``most likely'' trajectory from the competitive outcome to the cartel outcome, which helps us understand the mechanism driving collusive outcomes.

As in Figure \ref{fig-sim-prices}, we find that in the algorithmic market both prices and profits are higher on average than under competition, but the time paths display significant fluctuations. Apparently collusive episodes of high prices are followed by a series moves where firms undercut the opponent's price. These outcomes are reminiscent of the repeated game punishment strategies in price wars like \citet*{GreenandPorter84} or cartel equilibria in \citet*{APS86}.  \citet*{CalvanoetalAER2020} suggested that the collusive outcomes in their model were sustained by similar punishment strategies. However the mechanism supporting high prices is quite different in our environment, where collusive outcomes and reversion to competitive outcomes emerge endogenously as a result of learning by firms that are unable to implement complex repeated game strategies.

In our model firms' prices are driven by independent shock processes. However
the firms will get correlated shock realizations with a small positive probability.   As firms update their reaction function estimates, they will start to believe that their opponents will respond in kind to increases in prices. The perceived changes lead them to jointly increase prices, which reinforces their perceptions that the competing firm reacts to its own price on a one-to-one basis. This process leads the firms to rapidly boost prices to collusive levels near the joint monopoly outcome.  At the collusive outcome, the feedback process stops and the correlation between the two prices vanishes. Eventually, the independent idiosyncratic shocks pull them back to the Nash outcomes.  Together, this produces a pattern of collusion and reversion similar to the paths of play from repeated game strategies, but entirely driven by the learning dynamics.

Beyond our application to algorithmic pricing, our analysis highlights the important theoretical point that characterizing convergence and average limit outcomes are not sufficient to characterize behavior, as emphasized by \citet*{Sargent99}. Averaging over stochastic paths misses some important structure and leads to different predictions. \citet*{BanchioandMantegassa2023} and \citet*{Possnig2023} both use related stochastic approximation results to characterize the convergence of reinforcement learning rules to study the emergence of collusive outcomes.
\citet*{BanchioandMantegassa2023} study the fluid limit of
two interacting myopic reinforcement learners. Similar to our model, they emphasize how independent experimentation by players can produce correlated outcomes through the algorithmic dynamics. However their mechanism and limit results are different. They show that if the algorithms do not employ sufficient experimentation, they may converge to limits that sustain non-competitive outcomes. \cite*{Possnig2023} studies the stability of interactions of agents using a more sophisticated actor-critic algorithm. He shows that by expanding the state space that firms condition on, firms in a Cournot duopoly may converge to a collusive outcome. Since both of these papers focus on the deterministic limit dynamics, they cannot explain the recurrent stochastic dynamics observed in simulations.

The rest of the paper is organized as follows. In the next section we lay out the baseline model and the choice of model specification, all in a static environment with fixed beliefs. Section 3 then introduces algorithm and its dynamics. We characterize the limit outcomes, showing that all specifications converge to the competitive outcome. However the different model specifications have vastly different time series outcomes, which we illustrate and discuss in Section 4, which motivates the rest of the paper. Section 5 then turns to the analysis of the model, presenting our main analytic results. We provide a full formal analysis of belief updating and model specification, showing providing precise analytic results characterizing the limit outcomes and large deviation properties, and showing that the weight on the model with feedback converges to one. Section 6 delves further into the details underlying these results, analyzing dynamics
assuming that firms maintain a fixed, symmetric specification. Here we show that while the limit outcomes of the specifications are the same, the large deviation properties are very different.  Section 7 provides a brief conclusion. The appendix contains intermediate calculations, formal proofs, background results, and further illustrations.

\section{Baseline Model and Specification Choice}

We study a simple environment with strategic interaction, a textbook Bertrand duopoly model. In this environment we suppose that the firms face uncertainty about their opponents' strategies, that we parameterize into different specifications. Later we develop an algorithmic approach to pricing, which considers both specification choice and updating.

\subsection{Basic Setup}
We start with a Bertrand duopoly market with differentiated products where two firms are myopic optimizers.
The demand for good $i\in\{1,2\}$ produced by firm $i$ is
\begin{equation}
q_i = A - Bp_i + Cp_j \qquad i\ne j\in\{1,2\}
\label{eq: demand curve}
\end{equation}
where $A,B,C>0$.   Since $C>0$, two goods are strategic complements. We assume that $B-C>0$ (that is, the demand of good $i$ is more elastic to the price of good $i$ than the price of good $j\ne i$) to make the Nash equilibrium and the cartel outcome well defined. The best response $b_i(p_j)$ of firm $i$ against firm $j$'s price $p_j$ is
\begin{equation}
  b_i(p_j) =\frac{A+Cp_j}{2B}
\label{eq: best response}
\end{equation}
The unique Nash equilibrium price is
\[
  p^N =\frac{2 A}{2B-C}.
\]
Let $\Pi^N$ be the Nash equilibrium payoff of a duopolist.

The second benchmark is the cartel outcome $(p^C_1,p^C_2)$ that solves the joint maximization problem:
\[
\max_{(p_1,p_2)}p_1q_1+p_2q_2.
\]
A simple calculation shows
\[
p^C_1=p^C_2=\frac{A}{2(B-C)}.
\]
Let $\Pi^C$ be the payoff from the cartel outcome.  Clearly we have
\[
\Pi^C > \Pi^N.
\]

\subsection{Model Specification}

While most of the previous learning and algorithmic pricing literature considers a fixed subjective model specification, we embed the specification choice as part of the interaction.  The behavioral assumption behind a Nash equilibrium is that firm $i$ treats firm $j$'s action as a constant and optimizes against the constant function:
\begin{equation}
p_j=\alpha^0_{i}.
\label{eq: perceived law 0}
\end{equation}
Let ${\mathcal M}^0$ be the set of all functions represented as \eqref{eq: perceived law 0}, where each firm perceives its opponent's price as an unobserved constant.
Given $\alpha^0_{i}$, firm $i$ chooses its best response $b^0_{i}$ by solving
\[
\max_{p_i}p_i(A-B p_i+ C\alpha^0_{i}),
\]
which gives the best response:
\[
b^0_{i} = \frac{A+C\alpha^0_{i}}{2B}.
\]
We let $\Pi^0_i$ denote the optimized profit for firm $i$ under model $\mathcal{M}^0$.

Alternatively, firm $i$ may perceive that firm $j(\ne i)$ reacts to firm $i$'s price as in \eqref{eq: best response}.   Let ${\mathcal M}^1$ be the set of all linear reaction functions
\begin{equation}
p_j =\alpha^1_{0i}+\alpha^1_{1i}p_i  \qquad %\alpha_{1,i}\ne 0, \
i\ne j\in\{1,2\}.
\label{eq: perceived law 1}
\end{equation}
We define the belief vector as $\alpha^1_i=(\alpha^1_{0i},\alpha^1_{1i})$.  Duopolist $i$ then chooses its best response $b^1_{i}=b(\alpha^1_{i})$ by solving
\[
\max_{p_i}p_i(A-B p_i +C( \alpha^1_{0i}+\alpha^1_{1i}p_i)).
\]
For any vector $\alpha_i=(\alpha_{0i},\alpha_{1i})$, we define the best response function:
\begin{equation}\label{eq: br}
 b(\alpha_i)   = \frac{A + C\alpha_{0i}}{2(B - C\alpha_{1i})}.
\end{equation}
Then we have $b^1_i=b(\alpha_i^1)$, and we can also interpret $b^0_i=b((\alpha^0_i,0))$.

Clearly ${\mathcal M}^0$ is a subset of ${\mathcal M}^1$, and each represents a plausible subjective specification.   $\mathcal{M}^0$ depicts a common derivation of Nash equilibrium which optimizes against a fixed strategy, while $\mathcal{M}^1$ depicts a common derivation in a duopoly setting, where the Nash equilibrium is at the intersection of firms' reaction functions.

We consider a game where each duopolist first chooses a specification from ${\mathcal M}_i\in\{{\mathcal M}^0,{\mathcal M}^1\}$ $i\in\{1,2\}$.  For any specification, each firm $i$ will employ a forecast $f_i$ of the other firm $j$'s price, as  $f_i = \alpha^0_i$ for ${\mathcal M}^0$ and $f_i=\alpha^1_{0i} +\alpha^1_{1i}p_i$
for $\mathcal{M}^1$, then choose the corresponding response $b(\alpha_i) \in \{b_i^0, b_i^1\}$. Our equilibrium concept, self-confirming equilibrium, is a situation where
each firm optimizes given a specification and associated beliefs about the other firm, and each firm's forecasts are correct in equilibrium.
\begin{definition}
  $({\mathcal M}_1,{\mathcal M}_2)$ is a self-confirming equilibrium specification, if ${\mathcal M}_1,{\mathcal M}_2\in\{{\mathcal M}^0,{\mathcal M}^1\}$, and   $\exists(\alpha_1,\alpha_2)\in({\mathcal M}_1,{\mathcal M}_2)$ such that:
\begin{eqnarray*}
b_i &=& b(\alpha_i) \ \text{ as in (\ref{eq: br})},  \qquad\forall i \in\{1,2\} \\
b_j &=&f_i(b_i) \qquad\forall i\ne j\in\{1,2\}.
\end{eqnarray*}
\end{definition}
To be an equilibrium, the price of firm $i$ must be an optimal choice given its perceived law of motion $p_j=f_i(p_i)$ about how firm $j$ responds to firm $i$, and the belief must be confirmed at the optimal choice.   The notion of equilibrium is weaker than Nash equilibrium, as we admit misspecified beliefs about the opponent's response outside the equilibrium price.

The problem is that a pair of specifications can admit multiple self-confirming equilibria.  In Appendix \ref{app: SCE} we show that for the specification $({\mathcal M}^1,{\mathcal M}^1)$, a self-confirming equilibrium is $\alpha^1_i=(p^N,0)$, which supports that Nash equilibrium outcomes. However the belief $\alpha^1_i=(0,1)$ is also a self-confirming equilibrium, and it supports the collusive outcome $(p^C,p^C)$.
We need a criterion to select a self-confirming equilibrium for each pair of specifications, which we turn to next. Further, the equilibria that we just constructed were not arbitrary. We show below that the  two symmetric equilibria with $(\alpha_{1,i}=0,p_i=p^N)$ and $(\alpha_{1,i}=1, p_i=p^C)$ are important reference points in the algorithmic dynamics that we analyze.

\section{Dynamics of the Algorithm}

In this section we describe the structure of the algorithmic pricing that the firms employ.  We suppose that, given a particular specification of the opponent's behavior, each firm learns in real time. Each firm sets its price based on its current perception of its opponent's behavior, then after
observing its opponents price, it updates its beliefs. Similar models of adaptive learning are widely used both as stability and selection criteria for equilibria and as models to explain observed behavior.\footnote{See \citet*{FudenbergandLevine98}, \citet*{EvansandHonkapohja00}, and \citet*{Sargent99} for an overview and examples.}

Instead of imposing a particular belief specification, we consider learning with the endogenous choice of model specification. Rather than selecting a single specification, we suppose that each firm engages in model averaging. Such techniques are widely used in applied and empirical work, and are common in economic forecasting.\footnote{See \citet*{Steel2020} for a recent overview and references.}  Facing model uncertainty, the duopolist wants to hedge against the wrong choice of the specification by mixing the two specifications. The weight on the different specifications will then be updated over time depending on average forecast profits. We study the endogenous dynamics of model weights on models and beliefs within in a model.

\subsection{Layout}\label{sec: layout}

As described above, each duopolist entertains two specifications for the behavior of the opponent.   Instead of committing to a particular specification, a duopolist averages the two specifications.
Let $\pi_{i,t}$ be the probability assigned to specification ${\mathcal M}^1$ by duopolist $i$ in period $t$, which duopolist $i$ updates according to the average forecast profit of ${\mathcal M}^1$ over the average forecast profit of ${\mathcal M}^0$.   Ultimately, duopolist $i$ assigns a larger probability to a specification forecasting a larger long-run average profit.

At the beginning of period $t$, the state of the algorithm of duopolist  $i\in\{1,2\}$ is given by the vector:
\begin{equation}\label{eq: state vector}
\left(\pi_{i,t-1};\alpha^0_{0i,t-1}; \alpha^1_{0i,t-1},\alpha^1_{1i,t-1};\Pibar^0_{i,t-1},\Pibar^1_{i,t-1}\right)
\end{equation}
Each firm's behavior is described as an algorithm that updates this vector recursively, where the dating convention captures the information available.
Here $\pi_{i,t}$ is the probability assigned to the hypothesis that the opponent's behavior is specified as $\mathcal{M}^1$, $\alpha^0_{0i,t}$ is the coefficient for the specification of $\mathcal{M}^0$
as in (\ref{eq: perceived law 0}),
$\alpha^1_{i,t}=(\alpha^1_{0i,t},\alpha^1_{1i,t})$ is the coefficient vector of the linear reaction function in $\mathcal{M}^1$ as in (\ref{eq: perceived law 1}),
and $\Pibar^k_{i,t}$ is the average payoff of duopolist $i$ at the end of period $t$
when duopolist $i$ is using specification ${\mathcal M}^K$ where $k\in\{0,1\}$.

The timing of the model is such that each firm enters with a belief specification and model weight, chooses its price, and then based on its observations, updates the parameters.  We denote $b^0_{i,t}=b(\alpha^0_{i,t-1})$ and $b^1_{i,t}=\alpha^1_{i,t-1}$ as the recommended actions by $\mathcal{M}^0$ and $\mathcal{M}^1$, respectively. Duopolist $i$ chooses the best response in period $t$ by mixing the recommendations according to $\pi_{i,t-1}$:
\begin{equation}
b_{i,t}=(1-\pi_{i,t-1})b^0_{i,t}+\pi_{i,t-1} b^1_{i,t}.
\label{eq: best response}
\end{equation}
From this target best response, we suppose that the actual price includes some small noise:
\[
p_{i,t}=b_i+\epsilon_{i,t}
\]
where $\epsilon_{i,t}$ is an i.i.d. white noise with mean zero and variance $\sigma^2_{i,t}$. The shock
$\epsilon_{i,t}$ induces experimentation in prices, with the degree of experimentation measured by the second moment. This randomness could also arise from cost shocks, which naturally lead to variation in prices.
To simplify the notation, we assume that $\sigma^2_{i,t}$ is independent over ${i,t}$, and
\[
\sigma^2=\sigma^2_{i,t} \qquad i\in \{1,2\},\ t\ge 1.
\]
Our main results remain valid as long as the ratio of $\sigma^2_{1,t}/\sigma^2_{2,t}$ is uniformly bounded.

After observing $(p_{1,t},p_{2,t})$, duopolist $i$ updates the coefficients of each specification. The coefficients for each specification under ${\mathcal M}^0$ and ${\mathcal M}^1$ are updated according to recursive least squares.    If duopolist $i$ opts for $\mathcal{M}^0$, he estimates the average price of duopolist $j\ne i$:
\begin{equation}
\alpha^0_{0i,t}=\alpha^0_{0i,t-1}+{\lambda_t}\left( p_{j,t}-\alpha^0_{0i,t-1} \right).
\label{eq: recursive Nash}
\end{equation}
Here $\lambda_t$ is the gain sequence, which represents the weight on new information. We mostly focus on the case where $\lambda_t=1/t$, which means that each firm believes its opponent is playing a fixed strategy. Such a setting means that the estimate is simply the sample mean of the opponent's price:
\[
\alpha^0_{i,t-1}=\frac{1}{t-1}\sum_{s=1}^{t-1}p_{j,s}.
\]
Some of our illustrations use the constant gain case where $\lambda_t =\lambda $ is a small positive constant.   A constant gain is appropriate when a firm believes its opponent's strategy drifts over time. The same basic techniques and characterizations apply for the decreasing and constant gain cases, and only the sense of convergence differs.

If duopolist $i$ opts for $\mathcal{M}^1$, he estimates the slope $\alpha^1_{1i,t}$ and the intercept $\alpha^1_{0i,t}$:
\begin{eqnarray}
  \left[
    \begin{matrix}
      \alpha^1_{0i,t}\\
      \alpha^1_{1i,t}
      \end{matrix}
    \right] & = &
      \left[
    \begin{matrix}
      \alpha^1_{0i,t-1}\\
      \alpha^1_{1i,t-1}
      \end{matrix}
    \right] +{\lambda_t} R^{-1}_{i,t-1}
      \left[
    \begin{matrix}
      1\\
      p^1_{i,t}
      \end{matrix}
    \right]
    \left( p_{j,t}-\alpha^1_{0i,t-1}-\alpha^1_{1i,t-1}p^1_{i,t}    \right)
  \label{eq: recursive least square}   \\
  R_{i,t}& = & R_{i,t-1}+{\lambda_t}\left( \left[ \begin{matrix}
1   &   p^1_{i,t} \\
p^1_{i,t} &  (p^1_{i,t})^2
\end{matrix}
    \right] -R_{i,t-1}\right). \nonumber
\end{eqnarray}
Here $R_{i,t}$ is an estimate of the second moment matrix of the regressors. In the specification ${\mathcal M}^0$. the only regressor is a constant, so we can dispense with $R_{i,t}$. But under ${\mathcal M}^1$, the second moment matrix adjusts the speed of update of the intercept and slope parameters, as more volatile regressors naturally convey less information.

We can calculate the profit forecast by specification ${\mathcal M}^0$
\[
\Pi^0_{i,t}= p^0_{i,t}(A -B p^0_{i,t}+ C p_{j,t})
\]
and by specification ${\mathcal M}^1$
\[
\Pi^1_{i,t}= p^1_{i,t}\left(A -B p^1_{i,t}+ C p_{j,t}\right)
\]
from which we update the average expected payoff
\begin{eqnarray*}
\Pibar^0_{i,t} & = & \Pibar^0_{i,t-1}+{\lambda_t} \left( \Pi^0_{i,t}-\Pibar^0_{i,t-1}  \right) \\
\Pibar^1_{i,t} & = & \Pibar^1_{i,t-1}+{\lambda_t} \left( \Pi^1_{i,t}-\Pibar^1_{i,t-1}  \right)
\end{eqnarray*}
and
\[
\pi_{i,t}=\pi_{i,t-1} +{\lambda_t} \left( {\mathbb I}\left( \Pibar^1_{i,t}>\Pibar^0_{i,t}\right)-\pi_{i,t-1} \right).
\]
This completes the specification of the updating process of the algorithm.

\subsection{Implications for Collusion}

Before we turn to analysis of the model, we note several of its important features. In particular, we have shut down all avenues for either explicit or implicit collusion.

First, there is no communication between firms, which would be the most obvious and overt way to implement a collusive agreement.  Here each duopolist behaves independently, without any explicit or implicit communication. Moreover, each firm chooses its model specification independently each period, and given its specification, employs independent exploration in setting its price.

In the repeated game oligopoly literature, including the classic works of \citet*{GreenandPorter84} and \citet*{APS86}, collusion can also be sustained in equilibrium by patient firms who implement rather complex strategies. \citet*{CalvanoetalAER2020} suggest that the collusive outcomes they found were due to the algorithms learning repeated game strategies. Our specification eliminates such possibilities in two ways.  First, each duopolist is myopic, maximizing its one-period payoff. This shortcuts the possibility of implementing repeated game strategies with punishment phases, which are key to sustaining collusion. On top of that, we consider an algorithmic specification with fairly minimal complexity. Each duopolist $i$ only remembers the few parameters in its state vector (\ref{eq: state vector}), and does not know anything about duopolist $j$'s parameters or algorithm at all. Our duopolists are not capable of handling a complex repeated game strategy, which is essential for sustaining implicit collusion.

Other avenues for collusion rely on rational expectations, where firms correctly their forecast opponent's behavior. But in our setting, duopolist $i$'s algorithm does not know the internal state of the algorithm of the competing firm. Consequently, duopolist $i$'s algorithm cannot predict how the competing firm $j$ would react if $i$ were to deviate. Thus we differ from the algorithmic collusion literature where algorithms imply a form of  commitment where firms can learn opponents' strategies, as in  \citet*{Salcedo2015} and
\citet*{LambaZhuk2023}, or where asymmetries in frequency or commitment to
lead to collusive outcomes, as in \citet*{BrownandMacKay2023}.

Even more, in our setting both firms subjective models are misspecified. Duopolist $i$ does not know the probabilistic properties of its opponent's price $p_{j,t}$, treating it as an exogenous stochastic process. Our notion of self-confirming equilibrium imposes weak, on-path restrictions on firms' beliefs, which hold in the limit of the learning process.

\subsection{Summary of Stability under Learning}

We now analyze the long-run behavior of the algorithmic interactions. Here we briefly summarize the key results on the limit of beliefs, which illustrates the ineffectiveness of standard stability conditions based on the analysis of the associated continuous time dynamics. We apply standard stochastic approximation results in the learning literature, providing more detail below.

Given a pair of specifications $({\mathcal M}_1,{\mathcal M}_2)\in \{ {\mathcal M}^0,{\mathcal M}^1\}^2$, let $\Pi^{k\ell}_{i,t}$ be the profit of firm $i$ in period $t$ if duopolist 1 chooses ${\mathcal M}^k$ and duopolist 2 chooses ${\mathcal M}^\ell$ where $k,\ell\in \{0,1\}$.

\begin{proposition}\label{prop-converge}   Suppose that firm $i$ chooses ${\mathcal M}^1$.   Then, firm $i$ uses the linear reaction function to forecast firm $j$'s response
  \[
p_j=\alpha_{0i,t}+\alpha_{1i,t}p_i
\]
by estimating the coefficients $(\alpha^1_{0i,t},\alpha^1_{1i,t})$, and
\[
\lim_{t\rightarrow\infty}(\alpha^1_{0i,t},\alpha^1_{1i,t})=(p^N,0)
\]
with probability 1.   Moreover, for any combination $({\mathcal M}^k,{\mathcal M}^\ell)$ of specifications, the average profit $\Pibar^{k\ell}_{i,t}$ of firm $i$ converges to the Nash equilibrium profit:
\[
\lim_{t\rightarrow\infty} \Pibar^{k\ell}_{i,t}=\Pi^N
\]
with probability 1.
\end{proposition}

As long as $\epsilon_{1,t}$ and $\epsilon_{2,t}$ are independent, the slope of the estimated reaction function must be zero on average.  If we select a self-confirming equilibrium, which is a stable stationary solution of the least square learning dynamics, the payoff matrix of the specification game is written in terms of the payoff of duopolist 1:
\begin{equation}
  \bordermatrix{      & {\mathcal M}^0 & {\mathcal M}^1 \cr
    {\mathcal M}^0    &   \Pi^N,\Pi^N  &  \Pi^N,\Pi^N \cr
    {\mathcal M}^1    &   \Pi^N,\Pi^N  &  \Pi^N,\Pi^N \cr
  }.
    \label{eq: unperturbed game}
\end{equation}

The stability of the mean dynamics of the learning algorithm is a widely used selection criterion.  Here stability under learning results reduces the set of the self-confirming equilibria to a unique equilibrium where the duopolists play the Nash equilibrium regardless of the specification. For any combination of specifications, the stable stationary payoff from the learning dynamics is the Nash equilibrium payoff.
The mean dynamics cannot explain why a firm chooses a particular specification of the response of the other firm.  We need a refinement of self-confirming equilibrium.

\section{Illustrations and Motivation}
%Below we study the dynamics where the specification, beliefs, and prices are all determined endogenously

While the limit of all combinations of specifications is the Nash outcome $p^N$, we saw above in Figure \ref{fig-sim-prices} that the algorithmic price dynamics display sample path properties that are quite different. The time series feature recurrent escapes from the Nash price $p^N$ to an outcome near the collusive price $p^C$. Here we illustrate these features and discuss the mechanisms driving these dynamics. While the limits of the symmetric specifications where both firms use $\mathcal {M}^0$ and both use ${\mathcal M}^1$ are the same, $({\mathcal M}^0,{\mathcal M}^0)$ and $({\mathcal M}^1,{\mathcal M}^1)$ imply very different ``non-limit outcomes.''  Below we show that these differences are due to different large deviation properties.

\begin{figure}
\centerline{\includegraphics[width=0.8 \textwidth]{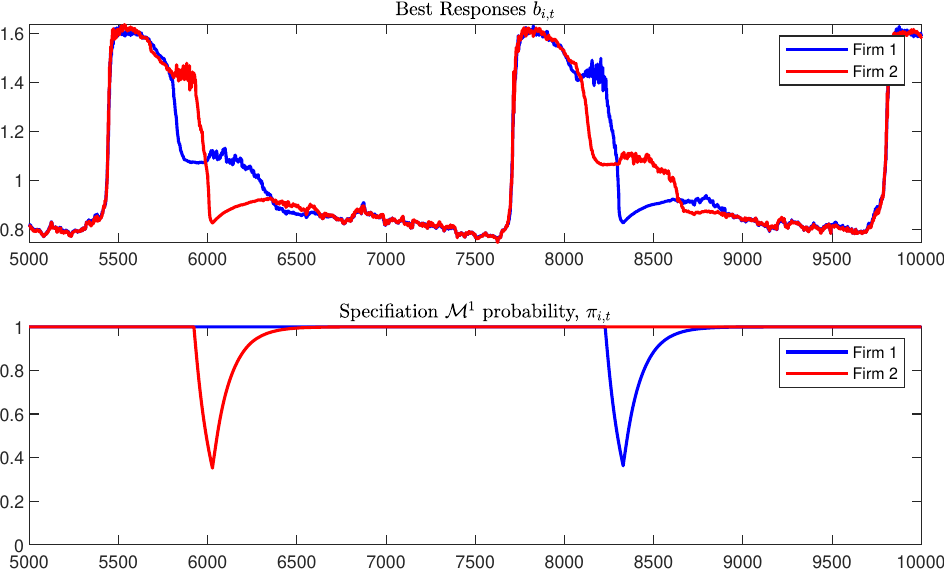}}
\caption{Top panel: simulated best responses $b_{i,t}$. Bottom panel: probability on specification $\mathcal{M}^1$, $\pi_{i,t}$.}\label{fig-simb}
\end{figure}

\subsection{Simulated Algorithmic Dynamics}

We have already seen the striking behavior of the algorithmic prices in Figure \ref{fig-sim-prices} above. Figure \ref{fig-simb} plots some of components underlying the price dynamics. For all of our illustrations, we use the following parameters: $A=1$, $B=1$, $C=0.7$.  With this parametrization, the Nash price is
$p_N = A/(2B-C)= 1/(2-0.7) \approx 0.77$, while the cartel price is $p^C=A/(2(B-C))=1/0.6\approx 1.67$.
We assume the price exploration shocks $\epsilon_{i,t}$ are i.i.d.\ normal, independent across firms, with mean zero and stand deviation $\sigma^2=0.025$. Later we explore the consequences of varying the degree of price exploration through $\sigma^2$, which has a key difference across the model specifications. Most of our plots use a constant gain $\lambda=0.01$. As we discuss below, as the gain $\lambda$ shrinks, the time between the episodes of price changes increases exponentially, consistent with convergence to the Nash outcome in the small gain limit.

The top panel of Figure \ref{fig-simb} plots the ``intended'' part of the price series, $b_{i,t}$, for the two firms in the same simulation as in Figure \ref{fig-sim-prices} above. This includes the averaged best responses across the two model specifications, $(1-\pi_{i,t})b^0_{i,t} + \pi_{i,t} b^1_{i,t}$, but not the exploration component $\epsilon_{i,t}$ of prices.  Thus it gives a less noisy version of the price dynamics that we saw in Figure \ref{fig-sim-prices} above. In both figures, we plot the last 5,000 periods of a 10,000-period simulation. The bottom panel of Figure \ref{fig-simb} plots the weight $\pi_{i,t}$ that each firm puts on the reaction function specification $\mathcal{M}^1$.

The convergence results in Proposition \ref{prop-converge} suggest that, at least in the long run, we would expect to see small fluctuations around the Nash equilibrium price. But instead we observe recurrent episodes where both firms increase their prices to collusive levels. The price increases occur roughly simultaneously over just a few periods from the Nash equilibrium price $p^N$ to near the collusive, joint monopoly price $p^C$.\footnote{The average price during these episodes falls just short of $p^C$, for reasons we will describe below.}  Prices remain near the joint monopoly levels for a while, then firms cut prices in a staggered fashion, and gradually both prices return to the Nash equilibrium levels.  Prices remain low for a while, until another episode where they increase prices again.

As we discuss below, the escapes from the Nash price $p^N$ to higher prices occur under the reaction-function specification $\mathcal{M}^1$, but not under the constant specification $\mathcal{M}^0$. Moreover these ``collusive episodes'' (without collusion) of higher prices are clearly beneficial to firms, allowing them to earn higher profits. Thus, as shown in the bottom panel of Figure \ref{fig-simb}, both firms put nearly all weight on specification $\mathcal{M}^1$, with $\pi_{i,t} \approx 1$ for most of the sample. The deviations from a model probability near unity occur during the asymmetric episodes of price cutting, with the second-moving firm finding temporarily higher profits under the constant $\mathcal{M}^0$ specification. Such asymmetric deviations are short-lived, and we show below that $\pi_{i,t}$ converges to one in the limit.

While the price dynamics are stochastic, they display clear and regular features. The timing of the episodes of price increases is random, and when the gain $\lambda$ is smaller, the dynamics slow and firms spend more time near the Nash level. But the collusive episodes recur and have similar character: each time, both firms act nearly in unison to increase prices to the same levels. This correlation in price increases is striking, as firms operate, choose algorithms, update, and experiment independently.

The correlation in prices breaks down once the firms approach the cartel price, and the price-cutting episodes are asymmetric. One firm initially undercuts the other, dropping part of the way back down toward the Nash price. While one firm initially remains at a higher price, it later responds by further undercutting down to near the Nash level. Figure \ref{fig-simb} shows that the price-cutting episodes are regular, but there is also a random component in determining which firm cuts price first.

\subsection{Differences Across Specifications}
\begin{figure}
\centerline{\includegraphics[width=0.8 \textwidth]{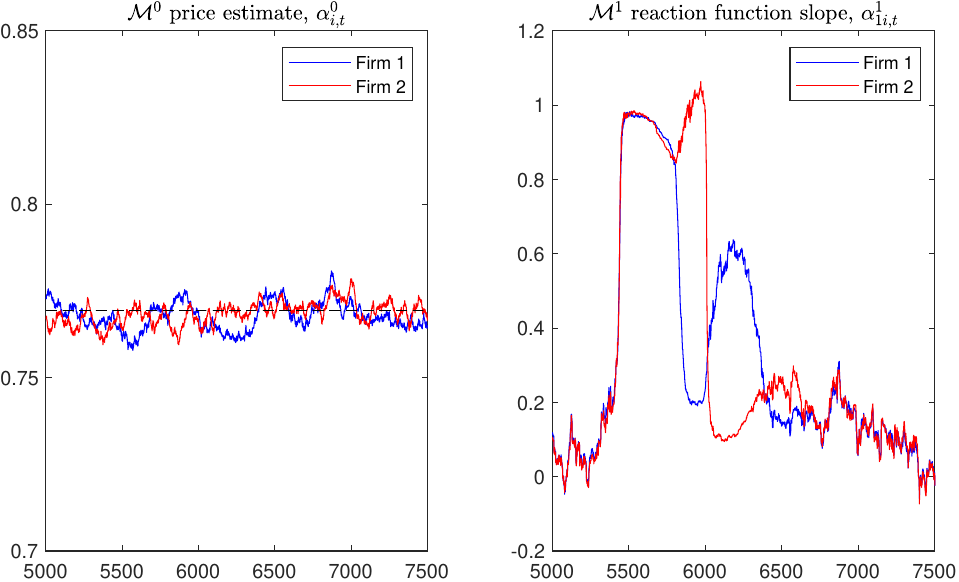}}
\caption{Left panel: $\mathcal{M}^0$ predicted price $\alpha^0_{i,t}$. Right panel: $\mathcal{M}^1$ reaction function slope $\alpha^1_{1i,t}$}\label{fig-sime}
\end{figure}

We find that the sample average payoff from ${\mathcal M}^1$ is larger than that from ${\mathcal M}^0$, even though both average payoffs converge to the same limit $\Pi^N$.
 Firms endogenously and independently find that entertaining a reaction function specification leads to higher profits, which in turn increases the probability weight assigned to ${\mathcal M}^1$.   In the end, $\pi_{i,t}\rightarrow 1$, so the price dynamics are dictated by $({\mathcal M}^1,{\mathcal M}^1)$, generating recurrent excursions from the Nash to the collusive outcome.

We now describe in more detail the dynamics of beliefs and prices in each specification.
Figure \ref{fig-sime} shows the time series of the key components of beliefs in the two specifications. Here we focus on the first half of the sample from the figures above, to isolate a single episode of price increases and decreases. The left panel of the figure shows the time series of the estimated coefficients $\alpha^0_i$ under $\mathcal{M}^0$ for each of the two firms. When each firm uses $\mathcal{M}^0$, the estimated price exhibits small fluctuations around the Nash equilibrium
price $p^N$, shown with a dashed line.  Note that the scale is smaller here than on the price dynamics plots above in order to make the belief changes visible. Since each firm estimates that its opponent plays close to $p^N$, the firms' best responses $b^{0}_{i,t}$ will also be close to $p^N$, and most price fluctuations come from the price shocks $\epsilon_{i,t}$ rather than revisions of beliefs.
There also appears to be no correlation in the deviations of firms' beliefs from the equilibrium levels, so prices would be uncorrelated as well. When the gain $\lambda$ shrinks, the distribution of beliefs becomes more tightly concentrated on the Nash equilibrium price.

However very different outcomes emerge when each firm uses the reaction-function specification $\mathcal{M}^1$. Firms now estimate an intercept and a slope, and we saw
in Proposition \ref{prop-converge} that $\alpha_{i,t}$ converges to $(p^N,0)$. The right panel of Figure \ref{fig-sime} shows the slope coefficients $\alpha^1_{1,it}$ for each of the two firms.  These are the key parameters which drive the price dynamics in Figures \ref{fig-sim-prices} and \ref{fig-simb} above. There are large fluctuations
in these estimated slopes, which increase in near unison from near zero to
near one during the episodes of price increases, then decline asymmetrically
and non-monotonically back to zero during the episodes of price cuts.

\subsection{Objectives}

The central idea of stochastic approximation, which we applied in Proposition \ref{prop-converge}, is to approximate the stochastic dynamics by deterministic dynamics represented by an ordinary differential equation. Relying on a law of large numbers, this approximation keeps track of the dynamics of the mean of the stochastic process, thus called the mean dynamics.   The conventional learning literature focuses on the stable stationary point of the mean dynamics.   Our exercise reveals that as we approximate a stochastic process by deterministic dynamics, we ``wash out'' important information about the stochastic process.  While every combination of the specifications shares the Nash equilibrium outcome as the unique stable stationary point, the  analysis of ${\mathcal M}^0$ and ${\mathcal M}$ reveals that the sample path of prices under the two specifications show dramatically different properties.

In our analysis, we examine the center of the distribution over the sample path, namely the mean dynamics, as well as the tail of the distribution where a rare event, such as the escape from the neighborhood of the stable stationary outcome, occurs.
By using large deviation theory to analyze the tail portion of the probability distribution over sample paths, we explain several features the mean dynamics and its stability analysis cannot explain.

\begin{itemize}
\item Why are the escapes from the stable point significantly more frequent under
  $({\mathcal M}^1,{\mathcal M}^1)$ than $({\mathcal M}^0,{\mathcal M}^0)$?

\item Why does the escape path always point to the collusive outcome?

\item How can we explain the almost perfect coordination of the escape time and path between the two duopolists, even though the selection of the specification and the exploration are independent?
\end{itemize}

An essential feature of our model is that the specification of the model is endogenously determined. In conventional learning models, the specification of the model is imposed by an outsider, and the players inside of the model can learn about the parameters. Instead, we assume that each firm chooses a specification from $\{ {\mathcal M}^0, {\mathcal M}^1\}$ probabilistically, according to the average performance of each specification.

\subsection{Discussion of the Mechanism Driving the Dynamics}\label{sec: discussion}

\begin{figure}
\centerline{\includegraphics[width=0.6 \textwidth]{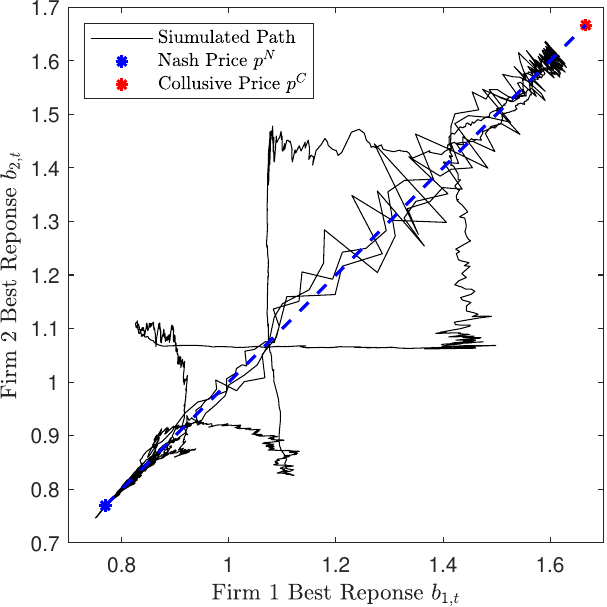}}
\caption{Simulated best responses for the two firms, $b_{1,t}$ versus $b_{2,t}$. Also shown are the Nash price $p^N$, the collusive price $p^C$, and the 45 degree line connecting them. }\label{fig-simf}
\end{figure}

Before proceeding to the analysis of the algorithm, we provide an overview and discussion of the key mechanism of the dynamics which drive our results.
Figure \ref{fig-simf} plots the best responses for the two firms, $b_{1,t}$  and $b_{2,t}$ as defined in \eqref{eq: best response}. In Figure \ref{fig-simb} above we saw these time series, but now we plot them against each other. The figure also shows the self-confirming and Nash limit outcome $p^N$ along with the collusive price $p^C$. The episodes of price increases are rapid, correlated movements nearly along the 45 degree line. The price-cutting episodes are slower, asymmetric movements where prices drop back down to $p^N$. The correlation and lack thereof are key elements we now describe.

We have seen that the key features of the time series are driven by the specification $\mathcal{M}^1$.  In the self-confirming equilibrium, which is the limit of learning, the slope of each firm's estimated reaction function is zero, so each firm believes that its opponent will set prices independently of its own price. With fixed beliefs, the firms' prices are just i.i.d.\ fluctuations around the  fixed best response function
$b(\alpha^1_i)$. Thus the equilibrium of the learning dynamics is the same in $\mathcal{M}^1$, when firms estimate a reaction function, as in $\mathcal{M}^0$ when they forecast their opponent's price as an average of past prices.

But even though the experimentation processes $\{\epsilon_{1,t}\}$ and $\{\epsilon_{2,t}\}$ are mutually independent, there will occasionally be correlated sequences of realizations.
These sequences are rare, and thus do not affect the mean of the limit distribution, as the events are ``washed out'' in calculating the mean dynamics by invoking a law of large numbers.  However tail events matter substantially and drive the key price dynamics, if the duopolists respond to the rare events which in turn reinforce the beliefs of the duopolists.   It takes a relatively short sequence of correlated shocks to trigger a sequence of changes in beliefs which are reinforced by the endogenous best response dynamics, leading to rapid movement away from the limit point.

In particular, since firms under $\mathcal{M}^1$ are estimating reaction functions, these correlated shock realizations -- which make prices move together -- will increase the perceived slope.   Once a firm starts to believe its opponent's price is be positively correlated with its own, $\alpha_{1,i}>0$, it will increase its price.  The perceived positive slope in the reaction function effectively makes the firm's residual demand less elastic. Firm $i$ thinks that firm $j$ will respond to an increase in $p_i$ with an increase in $p_j$. Then firm $i$'s perceived demand will be less responsive to its own price, and it will optimally increase its price.

Thus the correlated realizations lead to an increase in the perceived slope $\alpha^1_{1,it}$ and a corresponding increase in the expected price $b(\alpha^1_{i,t})$.  The increase in the average price is a notable deviation from recent observations, which were mostly clustered around the Nash price. The new data point is an influential observation affecting the firm's estimated reaction function. Thus instead of just the random movements in $\epsilon_{i,t}$, both firms increase $b^1_{i,t}$.  Firms' prices increase together, and they rapidly adjust their estimated reaction function to account for this observed correlation, which leads to further price increases. Thus the initial movement in beliefs is locally self-reinforcing, with an increase in perceived correlation leading to correlated behavior, which takes firms farther away from the Nash limit outcome.

As the price increases continue and approach the joint monopoly price $p^C$, the price increases slow down. Once firms' reaction function slopes are $\alpha^1_{i,t} =1$ they set the collusive price $p^C$, but at that point do not further increase prices. With little or no change in the expected price $b(\alpha^1_{i,t})$, the realized price movements then come to be driven more by the independent shocks $\epsilon_{i,t}$.

Thus the perceived correlation starts to decline, and so do prices. The
escape dynamics leading to the price increases are symmetric, driven by correlated behavior, but the price cuts which lead back to the Nash equilibrium level are asymmetric, with firms undercutting each other. After remaining near the joint monopoly price for a while, both firms start to cut prices slightly, then one undertakes a larger price cut. The other firm keeps its price high for a while, then undercuts its opponent to a lower price, which is later matched or further undercut.  The path of price cuts is asymmetric, and which firm starts process with the larger initial cut is random. Eventually, the price correlations decay back to zero, and prices approach the Nash equilibrium level.

As we have emphasized, our model shuts down all avenues for collusion between the firms.
The episodes of ``collusive outcomes'' with high prices are ``non-limit events'' driven by correlated shock realizations and the strategic complementarity of the environment. All that is required is that firms simply be minimally reactive to their observations of their opponent's behavior, which allows them to realize higher profits.

The remainder of the paper provides analytic backing to the description above.  We show that the specifications $\mathcal{M}^0$ and $\mathcal{M}^1$ have the same limit outcome but very different large deviation properties. For the specification $\mathcal{M}^1$, small deviations from the limit equilibrium -- in a particular direction, toward price coordination -- are self-reinforcing, and lead to the recurrent episodes of joint price increases.

\subsection{Perturbed Game}

Before proceeding to the analysis, we consider one further illustration, which sheds light on the dynamics of the probability weight assigned to ${\mathcal M}^1$. We have seen that in the limit all specification choices give rise to the same expected payoffs. But since firms' specification weights evolve endogenously along a sample path, what matters is not the limit \eqref{eq: unperturbed game}, but  a perturbed version of  obtained by replacing the payoff vector with the sample average.
\begin{equation}
  \bordermatrix{      & {\mathcal M}^0 & {\mathcal M}^1 \cr
    {\mathcal M}^0    &   \Pibar^{00}_{1,t},\Pibar^{00}_{2,t}  &  \Pibar^{01}_{1,t},\Pibar^{01}_{2,t} \cr
    {\mathcal M}^1    &   \Pibar^{10}_{1,t},\Pibar^{10}_{2,t}  &  \Pibar^{11}_{1,t},\Pibar^{11}_{2,t} \cr
  }
    \label{eq: perturbed game}
\end{equation}
Since $\Pi^{k\ell}_{i,t}\rightarrow\Pi^N$ $\forall k,\ell\in\{0,1\}$ and $\forall i\in\{1,2\}$, it is not straightforward to differentiate one average payoff from another.  Nevertheless, the numerical simulations reveal a useful pattern as we've seen. Evaluating the average payoffs in the simulations leads to:
\begin{equation}
\bordermatrix{    & {\mathcal M}^0    &  {\mathcal M}^1 \cr
{\mathcal M}^0    &   0.5913, 0.5913  &  0.5909, 0.5911 \cr
{\mathcal M}^1    &   0.5911, 0.5909  &  0.6796, 0.6796 \cr
  }
  \label{eq: perturbed example}
\end{equation}
While ${\mathcal M}^1$ is not a dominant specification over ${\mathcal M}^0$, it is a risk dominant specification with a very small $p$ value in the long run.   We shall formalize the intuition.

We have seen above that the pair of specifications $({\mathcal M}^0,{\mathcal M}^0)$ generates sample paths that stay in a small neighborhood of Nash equilibrium price. On the other hand, the pair of specifications $({\mathcal M}^1,{\mathcal M}^1)$ generates sample paths which entail repeated excursions to the collusive price with almost perfect coordination, returning to the Nash equilibrium in a less coordinated manner.  Thanks to the repeated excursions to the collusive outcome, the average profit for firm $1$ is strictly higher than the Nash equilibrium payoff $\Pi^N$.
\[
\Pibar^{11}_{1,t}\simeq 0.6796 > 0.5913 \simeq\Pibar^{00}_{1,t}\simeq \Pi^N.
\]
If firm $1$ uses ${\mathcal M}^0$, his payoff is mainly independent of duopolist 2's specification, remaining close to Nash equilibrium payoff $\Pi^N$.
\[
\Pibar^{00}_{1,t}\simeq 0.5913 \simeq 0.5909 \simeq \Pibar^{01}_{1,t}.
\]
By the same token, if firm 2 uses ${\mathcal M}^0$, duopolist 1's average profit changes little by switching from ${\mathcal M}^0$ to ${\mathcal M}^1$.   Consequently,
${\mathcal M}^1$ becomes the risk dominant strategy.  Since
\[
  \Pibar^{11}_{1,t}-\Pibar^{10}_{1,t}=0.0885 \gg -0.0004 =
  \Pibar^{01}_{1,t}-\Pibar^{00}_{1,t},
\]
it takes only a small probability of firm 2 choosing ${\mathcal M}^1$ to make ${\mathcal M}^1$ become the best response to duopolist 2's choice of (mixed) specifications.

\section{Analysis}\label{sec: analysis}

We use stochastic approximation to analyze the asymptotic limits and the large deviation properties of the endogenous variables in our model. In this section we establish three main analytic results:
\begin{enumerate}
\item under specification $\mathcal{M}^0$ each firm's expected profits are asymptotically equivalent to the Nash profits, while
\item under $\mathcal{M}^1$ the firms will recurrently visit a neighborhood of the collusive outcome and so earn higher profits, which together imply
\item in the long run both firms will put probability one on specification $\mathcal{M}^1$.
\end{enumerate}
We state the main results in the text, but the proofs and some supporting calculations are rather involved and are reported in Appendix \ref{app: analysis}.

\subsection{Large Deviations}
While we work with different processes, we take reaction function coefficients as an example.  Given a discrete time stochastic process $\{\alpha^1_{i,t}\}_{t=1}^\infty$, we construct a continuous time process by linearly interpolating the values at $t$ and $t-1$.  Since we are interested in the tail behavior of the algorithm, as described in the appendix we work with the left-shifted continuous time process $\alpha^{1,K}_i(\tau)$, which is the right tail of the interpolated sample path of $\{\alpha^1_{i,t}\}$ starting from time $K$.

We saw above that in the limit the algorithm converges to the Nash outcomes.
But we need to go beyond the convergence and stability results.  The focus of our investigation is the probability that the parameter estimates escape from a small neighborhood of stable point.  To do this, we use large deviation theory, which characterizes events with  limit probability zero.  More detail on the background and calculations is provided in Appendix \ref{app: LD background}, and in Section \ref{details-M0} below.

For simplicity, we assume that
\[
\sigma^2_{1,t}=\sigma^2_{2,t}=\sigma^2 \qquad\
\forall t\ge 1.
\]
The main results of the paper remains valid if $\sigma^2_{2,t}\ne\sigma^2_{1,t}$, as long as the ratio is uniformly bounded away from zero and from above. We assume that the exploration probability does not decrease too quickly. We define $m(K+\tau)$ as the smallest number $T$ satisfying $\sum_{t=t_K}^T\lambda_t\ge\tau$.

\begin{assumption}
  $\forall \tau>0$,
  \[
\lim_{K\rightarrow\infty}\frac{\sigma^2_{m(K+\tau)}}{\sigma^2_{t_K}}=1.
  \]
\end{assumption}
Roughly, $\sigma^2_t$ can converge to 0 at the same rate as $1/t^\gamma$ for some $\gamma\le 1$. Let
\[
  \epsilonhat_{i,t}=\epsilon_{i,t}/\sigma_{i,t}
\]
be the normalized exploration probability.

Let $I: X\rightarrow [0,\infty]$ be a non-negative extended real valued function.   In case of $\alpha^1_{i,t}\in {\mathbb R}^2$, $X\subset {\mathbb R}^2$.   Let
\[
\epsilonbar_{i,t}=\frac{1}{t}\sum_{k=1}^t\epsilonhat_{i,t}
\]
be the average of $\epsilonhat_{i,t}$. To avoid technical issues, we impose mild regularity conditions satisfied by virtually all popular distributions.

\begin{assumption}
  $\epsilonhat_{i,t}$ has a good rate function if for any closed subset $F$ and for any open subset $G$ of Borel set $X$,
  \begin{eqnarray}
-\inf_{x\in F}I(x) & \ge & \limsup_{t\rightarrow\infty}\frac{1}{t}\log \Prob \left( \epsilonbar_{i,t}\in F \right)  \label{eq: upper bound} \\
   & \ge &  \liminf_{t\rightarrow\infty}\frac{1}{t}\log \Prob \left( \epsilonbar_{i,t}\in G \right) \ge -\inf_{x\in G}I(x).  \label{eq: lower bound}
  \end{eqnarray}
  In addition, we assume that $\inf_{x\in G} I(x)<\infty$ and $\inf_{x\in F}I(x)<\infty$.
\end{assumption}

This is an example of a large deviation statement, providing exponential bounds on rare events.  We do not assume that $I(x)$ is continuous, although most popular distributions have continuous rate functions.   Instead, we require that $I(x)$ is a finite good rate function.

The core of the exercise is to calculate the rate function $\inf I(x)$ and the sample path that achieves the inequality as an equality in \eqref{eq: upper bound} and \eqref{eq: lower bound}.  We extend this result, which is a large deviation principle for  a random variable, to cover the sample paths of the parameters of the algorithm.  We are particularly interested in the probability of escape from a $\mu$-neighborhood $\mathcal{N}_\mu$ of the stable stationary point $\bar \alpha=(p^N,0)$. For such events, the rate function can be characterized as:
\[
\bar S(\mu)
=-\lim_{t\rightarrow\infty}\frac{1}{t}\log\Prob \left(
  (\alpha^1_{1,\tau},\alpha^1_{2,\tau})\not\in {\mathcal N}_\mu (\bar \alpha, \bar \alpha) \text{ for some $\tau \le t$} \right)
\]
when the limit is well defined.  Otherwise, we take $\limsup_{t\rightarrow\infty}$.

We use a simple result to check the rate function of a special class of random variables.
\begin{lemma}\label{lemma-infinite}
Suppose a sequence of bounded random variables $x_t$ with $\Expect x_t=0$ converges to 0 with probability 1.   Let ${\overline x}_t=\frac{1}{t}\sum_{s=1}^t x_s$ be the sample average.   Then, the rate function of ${\overline x}_t$ is infinite:
  \[
\limsup_{t\rightarrow\infty}\frac{1}{t}\log \Prob \left( \epsilonbar_{i,t}\in F \right)  = -\infty.
  \]
\end{lemma}

\subsection{Properties of $\Pibar^0_{i,t}$ and $\Pibar^1_{i,t}$}
In this section we show that the expected profits under specification $\mathcal{M}^0$ are ``essentially'' asymptotically equal to the Nash equilibrium
profits.  To do so, we make use of the following notion of asymptotic equivalence of random variables, following \citet*{DemboandZeitouni98}.
\begin{definition}
A random variable $x_t$ is exponentially equivalent to $0$ if :  $\forall\mu>0$,
\[
  -\lim_{t\rightarrow\infty}\frac{1}{t}\log\Prob\left(\left| x_t\right| >\mu \right) =\infty.
\]
\end{definition}
If a random variable is exponentially equivalent to a number, we can treat the random variable as equal to the number in the long run in any practical sense.   Abusing notation, if $x_t$ is exponentially equivalent to $m$, we write $x_t=m$ because we deal with a large $t$.

A numerical simulation shows that under $\mathcal{M}^0$, prices were very close to the Nash equilibrium level $p^N$. The following result, proved in Appendix \ref{app: pi0} makes the intuition precise.

\begin{lemma}\label{lemma-pi0}
$\Pibar^0_{i,t}$ is exponentially equivalent to $\Pi^N$.
\end{lemma}

By contrast, under $\mathcal{M}^1$, prices recurrently visit a neighborhood of the collusive price level $p^C$. In Appendix \ref{app: pi1} we construct an ``unusual'' sequence of $(\epsilon_{1,t},\epsilon_{2,t})$ where the two shocks are almost perfectly positively correlated.  While rare, such an event has a small but positive probability, and leads firms to increase their slope coefficient estimates  $\alpha^1_{11,t}$ and $\alpha^1_{12,t}$ simultaneously at the same rate. With a sufficient increase, the beliefs then enter a region where the expected direction of motion leads away from
the stable point. As firms start to move prices together, they increase the estimated slope in the reaction function, which leads to a further increase in prices. What is crucial for such a result to persist is that the threshold to enter the ``self-reinforcing'' region is small enough. In Lemma \ref{lm: threshold} in Appendix \ref{app: pi1} we show that the threshold gets smaller as $\sigma^2$, or the size of the experimentation gets smaller.

In general, as $\sigma^2\rightarrow 0$, the size of exploration decreases, and it becomes more difficult from the small neighborhood of the stable stationary point. In the case of specification ${\mathcal M}^0$, the lower bound of the rate function goes to infinity as $\sigma^2\rightarrow 0$, which is a crucial step in showing that $\Pibar^0_{i,t}$ is exponentially equivalent to $\Pi^N$.  In contrast, $\bar S(\mu;\sigma^2)$, the large deviation rate function now emphasizing its dependence on the exploration variance, is uniformly bounded from above even if $\sigma^2\rightarrow 0$.

\begin{lemma}\label{lemma-pi1}
  \[
\limsup_{\sigma^2\rightarrow 0} \bar S (\mu;\sigma^2)<\infty
  \]
\end{lemma}
\subsection{Dynamics of $\pi_{i,t}$}

Recall that $\pi_{i,t}$ is the probability assessment of specification $\mathcal{M}^1$ at time $t$, which evolves according to
\[
\pi_{i,t}=\pi_{i,t-1}+{\lambda_t} \left( {\mathbb I}(\Pibar^1_{i,t}>\Pibar^0_{i,t}) -\pi_{i,t-1}\right).
\]
Since we know that $\Pibar^0_{i,t}$ is exponentially equivalent to $\Pi^N$, we can write the recursive form as
\[
\pi_{i,t}=\pi_{i,t-1}+{\lambda_t} \left( {\mathbb I}(\Pibar^1_{i,t}>\Pi^N)-\pi_{i,t-1} \right),
\]
whose associated ODE is
\[
{\dot \pi}_i = \Prob\left( \Pibar^1_{i,t}>\Pi^N \right) -\pi_i.
\]

Our key result, proved in Appendix \ref{app: pi}, shows that in the limit each duopolist will put probability one on specification $\mathcal{M}^1$.
This result is a consequence of our previous analysis, that profits are higher under the specification with feedback than the constant specification.

\begin{proposition} \label{pr: first main result}
For any $\sigma^2_t\rightarrow 0$,
\[
  \lim_{t\rightarrow\infty}\Prob\left( \Pibar^1_{i,t}>\Pi^N \right)=1.
\]
\end{proposition}

Thus each duopolist, who chooses prices and experiments independently, and independently chooses a model specification, will in the long run choose the specification $\mathcal{M}^1$. Under this specification, the two firms will recurrently and in almost perfect unison, increase prices to near the cartelized, joint monopoly level, as first observed in \citet{Williams00}.  We have shut down all avenues for explicit or implicit collusion, and yet the pricing algorithms recurrently lead to prices near collusive levels. Thus the observation of parallel price increases to supra-competitive levels cannot on its own be used as evidence of collusion.  In the remainder of the paper we
further explore the key features of the model and its structure which underlie these results.

\section{Analysis of Fixed Specifications}

While a key aspect of our model is the endogenous specification choice, it is instructive to analyze cases with fixed, symmetric specifications.  The comparison of  environments with $({\mathcal M}^0, {\mathcal M}^0)$ and $({\mathcal M}^1, {\mathcal M}^1)$ shows the fundamental difference in the stochastic properties of the two different pairs of specifications. For simplicity, we refer here to $\mathcal{M}^i$ as the symmetric pair $(\mathcal{M}^i,\mathcal{M}^i)$.  In addition to being simpler, studying a fixed specification allows us to draw on existing results from the learning literature, including \citet*{ChoWilliamsandSargent02} and \citet*{Williams2019}.

The contrast in the large deviation properties of the two cases highlights features that conventional analysis of the stability of the learning dynamics cannot capture.  While both specifications converge to the Nash outcomes, the convergence in specification $\mathcal{M}^1$ is less robust, in a sense we make clear.   Applying the large deviation results allows us to characterize the collusive outcomes which arise without collusion.

\subsection{Analysis of $\mathcal{M}^0$}\label{details-M0}

As described in Section \ref{sec: layout}, under ${\mathcal M}^0$, firm $i$ treats the opponent's action as an unobserved parameter and estimates the parameter by averaging the past actions.
Let $\alpha^0_{i,t}$ be the estimate of duopolist $i$ about duopolist $j$'s action based on $t$ periods.  Given $\alpha^0_{i,t-1}$, firm $i$ chooses its best response $b^0_{i,t}$ in period $t$, with the actual price incorporating the exploration shock.
Given $(p_{1,t},p_{2,t})$, duopolist $i$ receives profit $\Pi^{00}_{i,t}$ in period $t$, and updates its estimate of duopolist $j$'s price according to the recursive algorithm \eqref{eq: recursive Nash}.

As discussed in Appendix \ref{app:M0-conv}, as $\lambda_t \rightarrow 0$, we
can approximate the dynamics from (\ref{eq: recursive Nash}) by the solution of the ordinary differential equation:
\[
  {\dot \alpha}^0_i(t) =  b(\alpha^{0}_j(t))-\alpha^{0}_i(t)
\]
The stable point of this ODE system for $i,j=1,2$ is the limit point of beliefs. Thus we can show that \eqref{eq: recursive Nash} converges to the Nash equilibrium price, which is a unique stable stationary point of the differential equation:
\[
\lim_{t\rightarrow\infty}\alpha^0_{i,t}= p^N
\]
with probability 1, where
\[
p^N=\frac{A}{2B-C}.
\]
Thus as $t \rightarrow \infty$ the beliefs $\alpha^0_{i,t}$ will converge to the stable point where expected prices $b_{i,t}$ are equal to the Nash equilibrium price $p^N$ and observed prices $p_{i,t}$ are i.i.d. fluctuations around this constant price level. This is consistent with what we observed in Figure \ref{fig-sime}.

We need to go beyond the convergence and stability results, using large deviation theory. More detail on the background and calculations is provided in Appendix \ref{app: LD background}. We saw above in Lemma \ref{lemma-pi0} that profits in $\mathcal{M}^0$ are asymptotically equivalent to the Nash profits. Here we provide more detail and illustration of this result.

In earlier work \citet*{ChoWilliamsandSargent02} and \citet*{Williams2019}, we adapted results from large deviation theory due to \citet*{KushnerandYin03} to characterize this probability. In particular, under mild regularity conditions on $\epsilon_{i,t}$, there exists a real-valued function rate function $S$ for sample paths that satisfies
\[
-\inf_{x\in X} S(x) = \lim_{t\rightarrow\infty}\frac{1}{t}\log \Prob\left(  (\alpha^0_{1,\tau},\alpha^0_{2,\tau}) \in X  \text{ for some $\tau \le t$} \mid \ (\alpha^0_{1,0},\alpha^0_{2,0})=(p^N,p^N) \right)
\]
We are particularly interested in exit from sets of radius $r$, so we define
\[ {\bar S}(r) = \inf_{\{x: \|x -(p^N,p^N) \|=r\}} S(x). \]
If $\epsilon_{i,t}$ is a Gaussian random variable with mean 0 and variance $\sigma^2$, we can compute $\bar S(r)$ explicitly.\footnote{The calculations here are analogous to \citet*{FreidlinandWentzell98}.} First, we suppose that the beliefs follow mean dynamics but are subject to a perturbation:
\begin{equation} \label{perturbed ODE0}
\dot \alpha^0_i(t) =  b(\alpha^{0}_j(t))-\alpha^{0}_i(t) + v_i(t).
\end{equation}
Then we stack the estimates into a vector $\alpha^0=(\alpha^0_1,\alpha^0_2)$ and define an instantaneous cost function $L$ that weights the perturbations $v$ by the shock variance:
\[  L (\alpha^0,v) = \frac{1}{2\sigma^2} v' v. \]
Here the cost function $L$ is independent of $\alpha^0$, but we keep the more general notation for later reference.
Then we solve the control problem:
\begin{equation*} %\label{Scrit}
  \bar{S}(r) = \inf_{v(\cdot), T}  \int_{0}^T
  L (\alpha^0(s),v(s))  ds
\end{equation*}
where the minimization is subject to the dynamics \eqref{perturbed ODE0}  and:
\begin{equation*}
 \alpha^0(0) =(p^N,p^N), \ \ \text{ and }  \|\alpha^0(T) -(p^N,p^N) \|=r .
 % \label{bounds}
 \end{equation*}
If $v \equiv 0$ then the beliefs follow the mean dynamics. The cost is zero, but the beliefs do not escape. To find the most probable escape path, we find a least cost path of
perturbations that pushes beliefs a distance $r$ from the Nash equilibrium.

As we discuss in Appendix \ref{app: LD background}, the solution of this cost minimization problem characterizes the most likely escape path from the stable stationary point. The probability of observing an escape on a bounded time interval is exponentially decreasing, at a  rate given by the minimized cost. The minimizing path from the cost function also determines the most likely escape path. With probability approaching one, escapes occur at the point determined from the solution of the cost minimization problem.

In Appendix \ref{app: LD M0} we show how to calculate the large deviation rate function $S(x)$ explicitly.  In particular, $S(x)$ is a quadratic function of $x$, with $S(p^N,p^N)=0$.  The level sets of the cost function are ellipses, with the major axis on the 45 degree line.  Thus we expect that large deviations from the equilibrium point are nearly symmetric, and are equally likely to result from correlated increases or decreases in prices.

Moreover, we can write:
\[ {\bar S}(r) =  \frac{S_0 \ r^2}{\sigma^2}, \]
for a constant $S_0>0$.  Note that as $r$ increases, the neighborhood of the stable stationary point becomes larger, which makes it more difficult for $(\alpha^0_{1,t},\alpha^0_{2,t})$ to escape.  If $\sigma^2$ decreases, the ``scope'' of exploring the prices around $b^0_{i,t}$ decreases, and it takes more explorations to escape from the neighborhood, as depicted in Figure \ref{fig-rates0}.  The comparative static properties of the large deviation rate function are central to our analysis.

\begin{figure}
\centerline{\includegraphics[width=0.7 \textwidth]{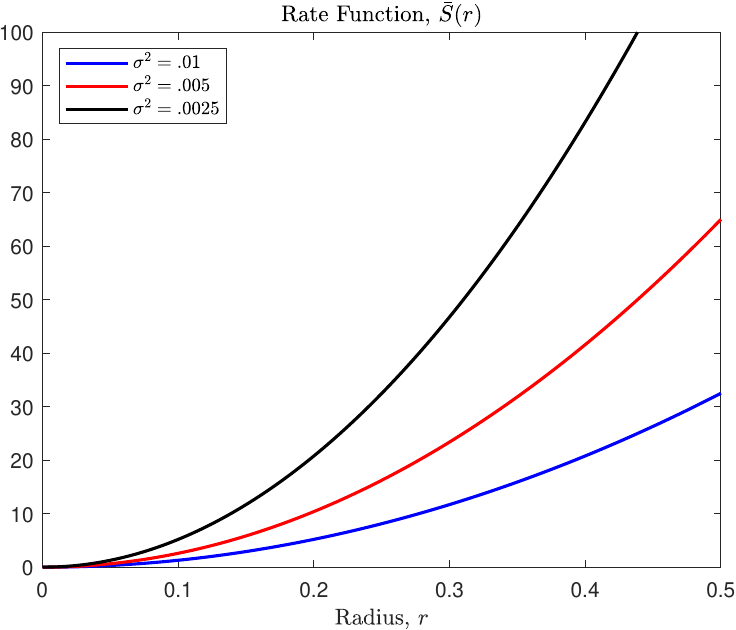}}
\caption{Solution of the cost minimization problem determining the rate function $\bar S(r)$ for different radius $r$ and different shock volatility $\sigma^2$ for the specification $\mathcal{M}^0$.}\label{fig-rates0}
\end{figure}

Overall, our results for $\mathcal{M}^0$ support what we observed in the simulations. Convergence to the equilibrium beliefs and Nash price is strong, and deviations are quite rare. When they occur, the most likely large deviations are associated with coordinated price movements. But decreases in prices are as likely to be realized as price increases, and the likelihood of significant changes in prices is very small. Furthermore, the rate function varies inversely with the variance $\sigma^2$, so for smaller shocks the likelihood of escaping the equilibrium diminishes rapidly. Many of these properties differ sharply when we allow feedback in the next section.

\subsection{Specification $\mathcal{M}^1$}

Now we suppose that each firm chooses  specification ${\mathcal M}^1$. As described above, duopolist $i$ predicts duopolist $j$ reacts to duopolist $i$'s price as a linear function (\ref{eq: perceived law 1}). Let $\alpha^1_{i,t}=(\alpha^1_{0i,t},\alpha^1_{1i,t})$ be the estimate of
firm $i$, who then chooses best response $b^1_{i,t}=b(\alpha^1_{i,t})$.
After observing $(p_{1,t},p_{2,t})$, each firm updates according to \eqref{eq: recursive least square}.

We again apply standard results in stochastic approximation to characterize the limit. Details of the calculations are in Appendix \ref{app:M1-conv}.  In particular, the mean dynamics ODEs now take the form:
\begin{eqnarray*}
\dot \alpha^1_i(t)  &=& R_i(t)^{-1} g_i(\alpha^1_i(t)) \\
\dot R_i(t) &=& M(\alpha^1_i(t) -R_{i}(t).
\end{eqnarray*}
Full expressions are given there, and in particular:
\[
  g_i(\alpha^1_i) =
\left[
    \begin{matrix}
      b(\alpha^1_{j}) -\alpha^1_{0,i} -\alpha^1_{1,i}b(\alpha^1_{i})   \\
      \left(b(\alpha^1_{j}) -\alpha^1_{0,i} -\alpha^1_{1,i}b(\alpha^1_{i})\right)b(\alpha^1_i) -\alpha^1_{1,i}\sigma^2
      \end{matrix}
    \right].
  \]
The stationary point of the ODE satisfies $g_i(\alpha^1_i)=0$, which thus requires  $b(\alpha^1_{j})=\alpha^1_{0,i}$ and $\alpha_{1i}=0$.  Therefore we have, for $i=1,2$,:
\[
\lim_{t\rightarrow\infty}(\alpha^1_{0i,t},\alpha^1_{1i,t})=(p^N,0)
\]
with probability 1, which is the Nash equilibrium outcome.   Thus, each duopolist's profit converges to $\Pi^N$.

While the equilibrium $(p^N,0)$ is stable, the mean dynamics which draw belief toward the equilibrium are ``weak'' along some directions. For example, two of the eigenvalues of the Jacobian matrix of the mean dynamics are very small, and vanish as $\sigma^2 \rightarrow 0$.  This is a manifestation of the multiplicity of self-confirming equilibria that we saw above.  Absent the price shocks, the equilibrium conditions $g_i(\alpha^1_i)=0$ do not uniquely identify the slope and intercept coefficients in the reaction function. Instead, any point $\alpha=(\alpha_0,\alpha_1)$ on the surface
\[
  \alpha_0=b(\alpha)-\alpha_{1}b(\alpha)
\]
will be an equilibrium.  With price shocks and hence $\sigma^2>0$, only $(p^N,0)$ satisfies the equilibrium conditions, but the other points along the surface come close to doing so, in that divergence of the mean dynamics from zero is of order $\sigma^2$.

We saw above that the episodes of price increases were symmetric, but there were asymmetries in the episodes when prices are cut from near the joint monopoly price $p^C$ to the Nash equilibrium price $p^N$. As discussed above, the price increases occur after an unusual sequence of correlated shocks. However, the price decreases are expected, and result from the mean dynamics leading back to the stable point. We illustrate the mean dynamics driving the price cutting episodes in Appendix  \ref{app:M1-conv}. Small differences in the realizations of shocks after the escapes determine which firm initiates the price cutting episode.

While both symmetric ${\mathcal M}^0$ and $\mathcal{M}^1$ specifications converge to the Nash equilibrium, the large deviation properties of the estimators are dramatically different.   As above, the rate function can be found by solving a cost minimization problem. The beliefs are initialized at the limit point $\alpha^1_i=(p^N,0)$ and they follow perturbed mean dynamics until they move a distance $r$.  In Appendix \ref{app: LD M1}, we show how to calculate the rate function when $\epsilon_{i,t}$ are Gaussian random variables with variance $\sigma^2$.

\begin{figure}
\centerline{\includegraphics[width=0.8 \textwidth]{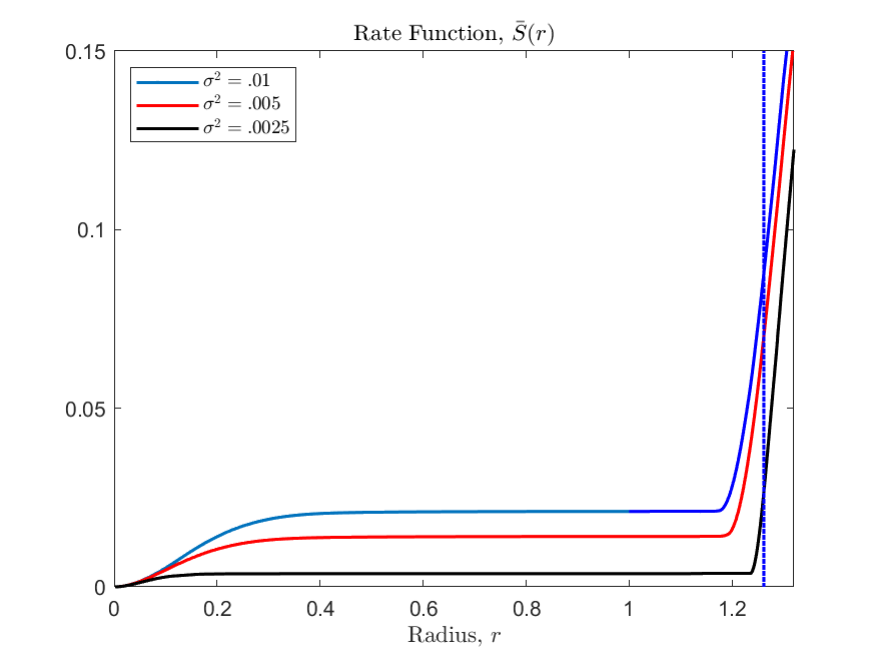}}
\caption{Solution of the cost minimization problem determining the rate function $\bar S(r)$ for different radius $r$ and different shock volatility $\sigma^2$ for the specification $\mathcal{M}^1$.
}\label{fig-rates1}
\end{figure}

The rate function $\bar S(r)$ for $\mathcal{M}^1$ is shown in Figure \ref{fig-rates1}.
We saw above that as $\sigma^2\rightarrow 0$, the probability of escaping from a small neighborhood of the stable stationary point under ${\mathcal M}^0$ decreases.    In contrast, the escape probability under ${\mathcal M}^1$ increases as $\sigma^2\rightarrow 0$.
As $\sigma^2$ converges to 0, so does $\bar S$.   For a small $\sigma^2$, it becomes very easy for $(\alpha^1_{1t},\alpha^1_{2,t})$ to escape from a ``small'' neighborhood of the Nash equilibrium outcome. This is a reflection of the weakness mean dynamics that we noted above. In order to escape from the equilibrium point, the perturbations must fight against the mean dynamics which would pull beliefs back to the equilibrium point. As the variance of the price shocks shrinks, the mean dynamics become weaker and the radius of local stability shrinks as well.

Note that ${\bar S}(r)$ is a weakly increasing function of $r$ but has a flat portion, implying that the size of the neighborhood where $(\alpha^1_{1t},\alpha^1_{2,t})$ can escape easily is in fact ``large.'' This explains the large fluctuations of the prices charged by each duopolist, which we've seen in Figures \ref{fig-sim-prices} and \ref{fig-simb} above. The shape of the rate function  ${\bar S}(r)$ means that once the beliefs escape a small distance from the equilibrium point, they are very likely to continue to escape until the inflection point where ${\bar S}$ starts to increase again.

\subsection{Relation Between Mean Dynamics and Large Deviations}

The significant escape dynamics and flat portion of the rate function are driven by features of the mean dynamics in the $\mathcal{M}^1$ specification.  Above we saw that the that there is a unique stable equilibrium which we can denote $\bar \alpha=[p^N,0]$, which is a stationary point of the mean dynamics $g(\bar \alpha)=0$ and all of the eigenvalues of the Jacobian matrix have negative real part. Thus at least locally around $\bar \alpha$ the mean dynamics point back toward $\bar \alpha$.

However, outside a neighborhood of $\bar \alpha$ the mean dynamics point away from it. We established this analytically in the proof of Lemma \ref{lemma-pi1} analytically, and we illustrate it numerically in the Appendix \ref{app: LD M1}.   This means that once the beliefs escape a neighborhood around the equilibrium, we would expect further movement away from the equilibrium.  The perturbations $v$ in the large deviation cost minimization problem are only necessary to push against the mean dynamics. If the mean dynamics point away from the equilibrium, further escapes happen ``for free'' with $v=0$.

We show in Appendix \ref{app: md-ldp} that these properties of the mean dynamics help explain the properties of the rate function $\bar S(r)$ and hence the escape dynamics.
$\bar S(r)$ is increasing for very small $r$ inside the stability radius, has a long flat portion where the mean dynamics lead the beliefs away from the stable equilibrium, and then starts increasing again when the mean dynamics point back toward the equilibrium.
For less volatile shocks with smaller $\sigma^2$, the radius of stability is smaller, so it is much easier to escape the equilibrium (the cost $S$ is much lower) and the inflection point is closer to the beliefs that support the collusive price.

As we discussed in Section \ref{sec: discussion}, the collusive outcomes arise in the absence of any opportunity for collusion. The collusive episodes are caused by  correlated realizations of independent shock processes, which are unlikely events with limit probability zero. Nonetheless, away from the limit, these collusive outcomes are a recurrent feature of the model.

\section{Conclusion}

We have  developed a model of algorithmic pricing which shuts down every channel for explicit or implicit collusion, and yet still generates collusive outcomes. Our model uses common, widely-used ingredients -- model averaging and least squares estimation -- in a simple, textbook setting -- linear Bertrand duopoly with myopic profit maximization.  In this environment, the algorithmic learning dynamics are stable, and we have asymptotic convergence to a competitive outcome. But even in such a common, stable, and well-behaved environment, recurrent collusive outcomes arise.

We show that the market experiences recurrent episodes where both firms set prices at collusive levels, followed by a period of price-cutting and reversion to competitive prices. We analytically characterize the dynamics of the model, using large deviation theory to explain the recurrent episodes of collusive outcomes. Our results suggest that collusive outcomes may be a recurrent feature of environments with complementarities, arising endogenously from interactions and adaptation, while not relying on explicit communication, tacit collusion, or punishment strategies supporting collusive outcomes. Our results thus provide a challenge for lawmakers and regulators seeking policy reforms to enhance competition in markets with algorithmic pricing. Even if outcomes appear highly collusive, an outsider cannot conclude that a firm learns to collude via an algorithm.

%\newpage

%\linespread{1.1}
%\bibliographystyle{plain}
\bibliographystyle{chicago}
\bibliography{adaboost}

%%%%%%%%%%%%%%%%%%%%%%%%%%%%%%%%%%%%
%%% APPENDIX STARTS HERE
%%%%%%%%%%%%%%%%%%%%%%%%%%%%%%%%%%%%

\appendix

\section{Construction of Self-Confirming Equilibria}\label{app: SCE}
As an example, consider a $({\mathcal M}^1,{\mathcal M}^1)$ pair of specifications as in (\ref{eq: perceived law 1}). The equilibrium conditions are then:
\begin{eqnarray}
  p_2 & = & \alpha_{0,1}+\alpha_{1,1}p_1 \label{eq: first 2} \\
  p_1 & = & \alpha_{0,2}+\alpha_{1,2}p_2 \label{eq: second 2} \\
  p_1 & = & \frac{A+C\alpha_{0,1}}{2(B-C\alpha_{1,1})} \label{eq: third 2} \\
  p_2 & = & \frac{A+C\alpha_{0,2}}{2(B-C\alpha_{1,2})} \label{eq: fourth 2}
\end{eqnarray}
Here \eqref{eq: first 2} and \eqref{eq: second 2} imply the perfect foresight conditions for firms 1 and 2, respectively.   \eqref{eq: third 2} and \eqref{eq: fourth 2} imply that the price of firms 1 and 2 must be best responses.
We now construct two reference self-confirming equilibria.
First, suppose that $\alpha_{1,1}=\alpha_{1,2}=0$. Then the equilibrium conditions reduce to:
\begin{eqnarray*}
  p_2 & = & \alpha_{0,1}, \ \   p_1  =  \alpha_{0,2} \\
  p_1 & = & \frac{A+C \alpha_{0,1}}{2B} , \ \
  p_2  =  \frac{A+C \alpha_{0,2}}{2B}.
\end{eqnarray*}
Which imply that the self-confirming equilibrium price is precisely the Nash equilibrium price
\[
p_1=p_2=p^N=\frac{A}{2B-C}.
\]

Alternatively suppose that
\[
\alpha_{0,1}=\alpha_{0,2}=0 \ \text{and} \ \alpha_{1,1}=\alpha_{1,2}=1.
\]
Then from (\ref{eq: third 2}) and (\ref{eq: fourth 2}) we get
\[
p_1=p_2=p^C
\]
and
\[
\Pi_1=\Pi_2=\Pi^C.
\]
Following the same logic, we can construct multiple self-confirming equilibria for $({\mathcal M}^0,{\mathcal M}^1)$ and $({\mathcal M}^1,{\mathcal M}^0)$.

\section{Proofs and Background for the Results in Section \ref{sec: analysis}}\label{app: analysis}
\subsection{Preliminaries}
It is necessary to define the basic notation for the analysis. We describe the notational convention of $\alpha^1_{i,t}=(\alpha^1_{0i,t},\alpha^1_{1i,t})$ as an example, which will be used for other endogenous variables such as $\alpha^0_{i,t}$.

Given a discrete time stochastic process $\{\alpha^1_{i,t}\}_{t=1}^\infty$, we construct a continuous time process by linearly interpolating the samples in $t$ and $t-1$ periods.
Define $\alpha^1_i(\tau)$ as the continuous time process obtained by linear interpolation. Since $\sum_{t=1}^T\lambda_t\rightarrow\infty$ as $T\rightarrow\infty$, $\forall K\in\{1,2,\ldots\}$, there is a unique $t_K$ that is the smallest $T$ where $\sum_{t=1}^T\lambda_t\ge K$.   For a real number $\tau>0$, define $m(K+\tau)$ as the smallest number $T$ satisfying $\sum_{t=t_K}^T\lambda_t\ge\tau$.   Given a continuous time process $\alpha^1_i(\tau)$, define the left shifted continuous time process as
\[
\alpha^{1,K}_i(\tau)=\alpha^1_i (K+\tau)
\]
which is essentially a right tail of the sample path of $\alpha^1_i(\tau)$ starting from
time $K$.   We will use the same convention for the other variables by adding superscript $^K$ and the real number $\tau>0$ to represent the left shift continuous time process.

\subsubsection{Exploration Probability}

$\pi_{i,t}$ is the probability that duopolist $i$ assigns to specification $\mathcal{M}^1$.  We assume that the duopolist is not committed to any particular specification.

\begin{assumption} \label{as: lower bound}
  $\exists \piubar,\pilbar\in (0,1)$ such that
  \[
0 < \pilbar \le \pi_{i,t} \le \piubar <1.
  \]
\end{assumption}

The following is the proof of Lemma \ref{lemma-infinite}.
\begin{proof}
  It suffices to show that the $H$ functional of the sample average
\[
  H({\overline x},\tau)=\lim_{\tau\rightarrow 0}
  \lim_{K\rightarrow\infty}\frac{1}{m(K+\tau)-t_K}\log\Expect \left[ e^{\alpha \sum_{t=t_K}^{m(K+\tau)}\xi_t} \ \mid \ {\overline x}_{t_K}=x\right]
\]
is 0.     We can write
  \[
{\overline x}_t={\overline x}_{t-1}+\frac{1}{t}\left( x_t-{\overline x}_{t-1} \right).
\]
Define $\xi_t=x_t-{\overline x}_{t-1}$ which is a bounded martingale difference.   Since $x_t\rightarrow 0$ with probability 1, $\forall\nu>0$, exists $T(\nu)$ such that $\forall t\ge T(\nu)$, $\abs{\xi_t}\le\nu$ with probability $1-\nu$.   Note that $\forall\alpha$, $\forall t_K\ge T(\nu)$,
\[
  e^{-\nu (m(K+\tau)-t_K)} (1-\nu) +e^{-M (m(K+\tau)-t_K)}\nu \le \Expect e^{\alpha\sum_{t=t_K}^{m(K+\tau)}
    \xi_t}  \le  e^{\nu (m(K+\tau)-t_K)} (1-\nu) +e^{M (m(K+\tau)-t_K)}\nu
\]
where $M=\sup_{t\ge 1}\abs{x_t} <\infty$ since $x_t$ is bounded.   We can choose $M'$ such that for any sufficiently small $\nu$,
\[
  e^{\nu (m(K+\tau)-t_K)} (1-\nu) +e^{M (m(K+\tau)-t_K)}\nu < e^{M' \nu( m(K+\tau)-t_K)}
\]
and therefore,
\[
-M' \nu\le \frac{1}{m(K+\tau)-t_K}\log\Expect e^{\alpha\sum_{t=t_K}^{m(K+\tau)}
\xi_t} \le M'\nu.
\]
Since $\nu$ is arbitrary, $H({\overline x},\tau)=0$ $\forall\tau>0$.
\end{proof}

Consider the left shift process $\alpha^{1,K}_i(\tau)$ initialized at  $\alpha^{1,K}_i(0)=(p^N,0)$ which is a stable stationary point of the recursive algorithm of $\alpha^1_{i,t}$.    Consider an open ball around the stable stationary point $(p^N,0)$ with radius $\mu$, ${\mathcal N}_\mu(p^N,0)$.   We are interested in the probability that $\alpha^{1,K}_i(\tau)$ escape from ${\mathcal N}_\mu(p^N,0)$, which is a rare event.
\citet*{DupuisandKushner89} proved that there exists a non-negative measurable function
$S(x,\tau,\phi)$ such that for any closed subset $F$ and open subset $G$ of ${\mathcal N}_\mu(p^N,0)$ such that
\begin{eqnarray*}
&& \limsup_{K\rightarrow\infty}a_{t_K}
   \log \Prob\left( \exists\tau>0, \ \alpha^{1,K}_i(\tau)\not\in{\mathcal N}_\mu(p^N,0) \ \mid \  \alpha^{1,K}_i(0)=(p^N,0) \right)  \nonumber \\
& \le &
        -\inf_{\substack{\phi(0)=(p^N,0), \\ \exists\tau',\phi(\tau')\not\in {\mathcal  N}_\mu(p^N,0)}}S(\alpha^{1,K}_i(0),\tau,\phi)\equiv S^*(\alpha^{1,K}_i(0),{\mathcal N}_\mu(p^Nj,0)).
  \label{eq: DK89}
\end{eqnarray*}
 We call $S^*(\cdot)$ the rate function.

\subsection{Properties of $\Pibar^0_{i,t}$}\label{app: pi0}

We know that $\alpha^0_{i,t}$ and $\alpha^1_{i,t}=(\alpha^1_{0i,t},\alpha^1_{1i,t})$ satisfy the large deviation properties (\citet*{DupuisandKushner89}).   Let $S^{\alpha,0}_i({\mathcal N}_\mu(0),\sigma^2)$ and $S^{\alpha,1}_i({\mathcal N}_\mu(0,0),\sigma^2)$ be the rate function around $\mu$ neighborhood of the stable stationary point of $\alpha^0_{i,t}$ and $\alpha^1_{i,t}$ in case $\sigma^2_{i,t}=\sigma^2$ $\forall t\ge 1$, respectively.   Recall that $b^0_{i,t}$ and $b^1_{i,t}$ are the continuous transformation of $\alpha^0_{i,t}$ and $\alpha^1_{i,t}$.    By the contraction principle of the large deviation, the large deviation rate functions of $b^0_{i,t}$ and $b^1_{i,t}$ are the continuous transformation of the large deviation rate functions of $\alpha^0_{i,t}$ and $\alpha^1_{i,t}$.   That is,
there exist continuous functions $g^0$ and $g^1$ such that
$\forall\mu>0$,
\begin{eqnarray*}
&& -\lim_{t\rightarrow\infty}\frac{1}{t}\log\Expect \Prob\left( \left| b^0_{i,t}-p^N\right|>\mu \right) \le g^0 (S^{\alpha,0}_i({\mathcal N}_\mu(0),\sigma^2)) \\
&& -\lim_{t\rightarrow\infty}\frac{1}{t}\log\Expect \Prob\left( \left| b^1_{i,t}-p^N\right|>\mu \right) \le g^1 (S^{\alpha,1}_i({\mathcal N}_\mu(0,0),\sigma^2))
\end{eqnarray*}
which imply
\begin{equation}
  \lim_{t\rightarrow\infty}b^0_{i,t}=p^N \ \ \text{and} \ \
  \lim_{t\rightarrow\infty}b^1_{i,t}=p^N   \qquad\forall i\in\{1,2\}
  \label{eq: bconverge}
\end{equation}
with probability 1.   Recall that
\[
p_{j,t}=(1-\pi_{j,t})b^0_{j,t}+\pi_{j,t}b^1_{j,t}+\epsilon_{j,t}
\]
with $\pi_{j,t}\ge \pilbar>0$ and
\begin{eqnarray}
  \alpha^0_{i,t}
  &=&\pbar_{j,t}=\frac{1}{t}\sum_{k=0}^{t-1}p_{j,k} \\
  &=&\sum_{k=0}^{t-1}(1-\pi_{j,k})b^0_{j,k}+\pi_{j,k}b^1_{j,k}+\frac{1}{t}\sum_{k=0}^{t-1} \epsilon_{j,k}.
\end{eqnarray}
Thanks to \eqref{eq: bconverge},
\[
\lim_{t\rightarrow \infty} (1-\pi_{j,t})b^0_{j,t}+\pi_{j,t}b^1_{j,t}=p^N
\]
with probability 1.   Therefore, $\forall\mu>0$,
\[
\lim_{t\rightarrow\infty}\frac{1}{t}\log\Prob\left(\left|
\frac{1}{t}\sum_{k=0}^{t-1} (1-\pi_{j,k})b^0_{j,k}+\pi_{j,k}b^1_{j,k}-p^N
\right| >\mu \right) =-\infty.
\]
Since
\[
  \frac{1}{t}\sum_{k=0}^{t-1} (1-\pi_{j,k})b^0_{j,k}+\pi_{j,k}b^1_{j,k}
\]
is exponentially equivalent to $p^N$ (\citet*{DemboandZeitouni98}), we can treat
\[
\alpha^0_{i,t}=\pbar_{j,t} =p^N +\frac{1}{t}\sum_{k=0}^{t-1}\epsilon_{j,k}  \qquad i\ne j\in\{1,2\}.
\]
Hence,
\[
b^0_{i,t}= p^N +\frac{C}{2B}\frac{1}{t}\sum_{k=0}^{t-1}\epsilon_{j,k}.
\]
Since $\epsilon_{j,t}$ satisfies the large deviation property, $\forall\mu>0$,  $\exists
S^{\epsilon}_i(\mu)$ such that
\[
  -\lim_{t\rightarrow\infty}\frac{1}{t}\log\Prob\left(\left|
b^0_{i,t}-p^N
\right| >\mu \right) \le S^{\epsilon}_i(\mu).
\]
Thus,
\[
\Pi^0_{i,t}=\max_p p(A-Bp+C\alpha^0_{i,t})=\Pi^N +D {\overline\epsilon}_{i,t}+D'({\overline\epsilon}_{i,t})^2
\]
for some constant $D, D'$.   Thus, if ${\lambda_t}=\frac{1}{t}$,
\[
\lim_{t\rightarrow \infty}\Pi^0_{i,t}=\Pi^N
\]
with probability 1.

With those preceding calculations, the following is the proof of Lemma \ref{lemma-pi0} in the text.
\begin{proof}
Since
\[
  \Pibar^0_{i,t}=\Pibar^0_{i,t-1}+{\lambda_t} \left( \Pi^0_{i,t}-\Pibar^0_{i,t-1}\right),
\]
$\Pibar^0_{i,t}$ is the average of $\Pi^0_{i,t}$ which converges to $\Pi^N$ with probability 1.  Thus, $\Pibar^0_{i,t}$ is exponentially equivalent to $\Pi^N$.
\end{proof}

\subsection{Properties of $\Pibar^1_{i,t}$}\label{app: pi1}

To examine the sample path properties of $\Pi^1_{i,t}$, we need to examine the properties of $\alpha^1_{i,t}=(\alpha^1_{0i,t},\alpha^1_{1i,t})$.  It is more convenient to replace intercept $\alpha^1_{0i,t}$ by $(\pbar_{1,t},\pbar_{2,t})$ for each $i$ to write the updating rule as
\[
p_{j,t}-\pbar_{j,t}=\alpha^1_{i,t}(p_{i,t}-\pbar_{i,t}) \qquad\forall i\ne j,
\]
from which we can recover the intercept according to
\[
\alpha^1_{0i,t}=\pbar_{j,t}-\alpha^1_{i,t}\pbar_{i,t} \qquad\forall i\ne j.
\]
Since
\[
  p_{i,t}=(1-\pi_{i,t-1}) \frac{A+C\alpha^0_{0,t-1}}{2B}+\pi_{i,t-1}
  \frac{A+C\alpha^0_{0,t-1}}{2(B-C\alpha^1_{1i,t-1})} +\epsilon_{i,t},
\]
Recall that
\[
\alpha^0_{0i,t}=\pbar_{j,t}.
\]
Instead of
\[
(\pi_{1,t},\pi_{2,t};\alpha^0_{01,t},\alpha^0_{02,t};\pbar_{1,t},\pbar_{2,t}; \alpha^1_{11,t},\alpha^1_{12,t}),
\]
we can investigate
\[
(\pi_{1,t},\pi_{2,t};\pbar_{1,t},\pbar_{2,t}; \alpha^1_{11,t},\alpha^1_{12,t}).
\]
Moreover, we know we can treat
\[
  \pbar_{i,t}=p^N + \frac{1}{t}\sum\epsilon_{i,t}.
\]
Thus, our analysis focuses on the evolution of the remaining four state variables
\[
(\pi_{1,t},\pi_{2,t}; \alpha^1_{11,t},\alpha^1_{12,t})
\]
with the initial value set at the stable stationary point.

Note
\[
\alpha^1_{i,t}=\frac{\sum (p^1_{ik}-\pbar_{ik})(p_{jk}-\pbar_{jk})}{\sum(p^1_{ik}-\pbar_{ik})^2}
\]
where
\[
  p_{jk}=(1-\pi_{jk})(b^0_{jk}-\pbar_{jk})+\pi_{jk}(b^1_{jk}-\pbar_{jk})+\epsilon_{jk}
\]
and
\[
p^1_{ik}=b^1_{ik}+\epsilon_{ik}.
\]
Since $b^0_{jk}=p^N$ and $\pbar_{jk}=p^N$,
\[
  p_{jk}=\pi_{jk}(b^1_{jk}-p^N)+\epsilon_{jk}.
\]
Around the Nash equilibrium outcome, $(\alpha^1_{0i,t},\alpha^1_{1i,t})=(p^N,0)$ and therefore, $b^1_{i,n}(p^N,0)=p^N$.   Hence,
\[
  \alpha^1_{1i,t}=\frac{\sum_k \epsilon_{1k}\sum \epsilon_{2k} }{\sum_k \epsilon^2_{ik}}
  \rightarrow 0,
\]
since $\epsilon_{1k}$ and $\epsilon_{2k}$ are mutually independent.   Hence,
\[
{\dot\alpha}^1_{1i}=-\alpha^1_{1i}
\]
confirming the stability of $\alpha^1_{1i}=0$.   Thus, an event like
\[
  \exists T\ge 1, \
(\alpha^1_{11,T},\alpha^1_{12,T})\in {\mathcal W}(\alphalbar,\betalbar)
\]
where  ${\mathcal W}(\alphalbar,\betalbar)$ is a convex hull of $\{(\alphalbar,\alphalbar), (\alphalbar,\betalbar), (\betalbar,\alphalbar)\}$ for $\betalbar > \alphalbar>0$.

\begin{lemma} \label{lm: threshold}
  $\exists\betalbar>\alphalbar>0$ such that if the initial condition
  $(\alpha^1_{11,1},\alpha^1_{12,1})$ is in ${\mathcal W}(\alphalbar,\betalbar)$, then $({\dot\alpha}^1_{11},{\dot\alpha}^1_{12})>0$.   Moreover, we can choose $\alphalbar$ and $\betalbar$ satisfying
  \[
    0 < \lim_{\sigma^2\rightarrow 0}\frac{\alphalbar}{\sigma^2} < \infty \ \ \text{and} \ \
    0 < \lim_{\sigma^2\rightarrow 0}\frac{\betalbar}{\sigma} < \infty.
  \]
\end{lemma}

\begin{proof}  Suppose that $(\alpha^1_{11,1},\alpha^1_{12,1})=(\alphalbar,\alphalbar)$ for some $\alphalbar$ whose value will be determined shortly.  Recall $\pbar_{1,t}=\pbar_{2,t}=p^N$.    Since $\pbar_{i,t}$ is an average of a function of
  $(\alpha^1_{11,1},\alpha^1_{12,1})$, $\pbar_{i,t}$ moves a slower time scale than
  $(\alpha^1_{11,1},\alpha^1_{12,1})$, and can be treated as a constant equal to $p^N$ while investigating the long run dynamics of $(\alpha^1_{11,1},\alpha^1_{12,1})$.

Note that with a small but positive, probability, $(\alpha^1_{11,t},\alpha^1_{12,t})\in {\mathcal W}(\alphalbar,\betalbar)$ can occur.   Consider an ``unusual'' sequence of $(\epsilon_{1,t},\epsilon_{2,t})$ where the two shocks are almost perfectly positively correlated.   While rare, such an event has a vanishingly small positive probability and leads to $\alpha^1_{11,t}$ and $\alpha^1_{12,t}$ increasing simultaneously at the same rate to reach a small neighborhood of $(\alphalbar,\alphalbar)$ while keeping $\pbar_{1,t}$ and $\pbar_{2,t}$ close to $p^N$.

Since $(\alpha^1_{11,t},\alpha^1_{12,t})=(\alphalbar,\alphalbar)$,
\[
b^1_{1,t} -\pbar_{1,t}=b^1_{2,t} -\pbar_{2,t}.
\]
Therefore,
\[
  \Expect (b^1_1-p^N)(b^1_2-p^N) =\Expect (b^1_1-p^N)^2.
\]
Note
\[
p_{j,t}=(1-\pi_{j,t})(b^0_{j,t}-\pbar_{j,t})+\pi_{j,t}(b^1_{j,t}-\pbar_{j,t})+\epsilon_{j,t}.
\]
Since $\pbar_{j,t}=p^N$,
\[
p_{j,t}=\pi_{j,t}(b^1_{j,t}-p^N)+\epsilon_{j,t}=\pi_{j,t}(b^1_{i,t}-p^N)+\epsilon_{j,t}.
\]
Since $\epsilon_{1,t}$ and $\epsilon_{2,t}$ are orthogonal,
\[
  \Expect (p_{1,t}-\pbar_{1,t})(p_{2,t}-\pbar_{2,t}) = \pi_j\Expect (b^1_i-p^N)^2.
\]
Similarly,
\[
  \Expect (p^1_{1,t}-\pbar_{1,t})^2 =  \Expect (b^1_1-p^N)^2+\sigma^2.
\]
Thus,
\[
  {\dot\alpha}^1_i=
  \frac{\pi_j (b^1_i-p^N)^2}{(b^1_i-p^N)^2+\sigma^2} -\alpha^1_i \qquad\forall i\ne j\in\{1,2\}
\]
which is positive if and only if
\begin{eqnarray*}
  0
  & < & (\pi_j-\alpha^1_{1i}) (C\alpha^1_{1i} p^N)^2 -4\alpha^1_{1i} \sigma^2 (B-C\alpha^1_{1i})^2 \\
  & = &\alpha^1_{1i}
  \left[C^2 (p^N)^2 (\pi_j-\alpha^1_{1i})\alpha^1_{1i} -4\sigma^2(B-C\alpha^1_{1i})^2 \right].
\end{eqnarray*}
Since $\alpha^1_{1i}>0$, the equation is positive if and only if the quadratic term inside of the bracket is positive.   Let $\alphalbar$ and $\alphaubar$ be the two solutions of the quadratic equation.
We can easily show that $\exists \sigmaubar^2$ such that $\forall\sigma^2\in (0,\sigmaubar^2)$, the two roots of the quadratic equations are real-valued, and
\[
\alphalbar\alphaubar =\frac{4\sigma^2 B^2}{C^2\left( (p^N)^2\pi_j +4\sigma^2 \right)}
\]
and
\begin{equation}
\lim_{\sigma^2\rightarrow 0}\alphaubar =\pi_j \ge \pilbar>0
\label{eq: alphaubar bound}
\end{equation}
by Assumption \ref{as: lower bound}.   Thus,
\begin{equation}
  \lim_{\sigma^2\rightarrow 0} \frac{\alphalbar}{\sigma^2}
  =\frac{4B^2}{\pi_j C^2\left( (p^N)^2\pi_j+4\sigma^2\right)} >0.
  \label{eq: alphalbar linear rate}
\end{equation}

Given the initial condition $(\alpha^1_{11},\alpha^1_{12})=(\alphalbar,\alphalbar)$, the associated ODE of $(\alpha^1_{11},\alpha^1_{12})$ is
\begin{eqnarray*}
  {\dot\alpha}^1_{11} &=&\frac{(b^1_1-p^N) (b^1_2-p^N)}{ (b^1_1-p^N)^2+\sigma^2} -\alpha^1_{11} \\
  {\dot\alpha}^1_{12} &=&\frac{(b^1_1-p^N) (b^1_2-p^N)}{(b^1_1-p^N)^2+\sigma^2} -\alpha^1_{12}.
\end{eqnarray*}
Since
\[
  b^1_i-p^N
  =\frac{A+C p^N(1-\alpha^1_{1i})}{2(B-C\alpha^1_{1i})}-p^N,
\]
\[
  {\dot\alpha}^1_{1i}=\frac{\alpha^1_{11}\alpha^1_{12}C^2(p^N)^2}{
    (\alpha^1_{1i})^2C^2(p^N)^2+\sigma^2}-\alpha^1_{1i} \qquad\forall i.
\]
We are interested in the dynamics of $(\alpha^1_{11,t},\alpha^1_{12,t})$ from the neighborhood of initial condition $(\alphalbar,\alphalbar)$ where ${\dot\alpha}^1_{1i}>0$ $\forall i$.   Since
\begin{eqnarray*}
  {\dot\alpha}^1_{11}-{\dot\alpha}^1_{12}
&  = & \frac{\alpha^1_{11}\alpha^1_{12}C^2(p^N)^2}{
    (\alpha^1_{11})^2C^2(p^N)^2+\sigma^2} -\frac{\alpha^1_{11}\alpha^1_{12}C^2(p^N)^2}{
    (\alpha^1_{1i})^2C^2(p^N)^2+\sigma^2}-\alpha^1_{1i}+(\alpha^1_{12}-\alpha^1_{11})  \\
& = & \frac{ \alpha^1_{11}\alpha^1_{12}C^2(p^N)^2}{    ((\alpha^1_{11})^2C^2(p^N)^2+\sigma^2)((\alpha^1_{12})^2C^2(p^N)^2+\sigma^2)}
        (\alpha^1_{12}-\alpha^1_{11})(\alpha^1_{12}+\alpha^1_{11}) +(\alpha^1_{12}-\alpha^1_{11}) \\
& = & (\alpha^1_{12}-\alpha^1_{11})
        \left[
\frac{ \alpha^1_{11}\alpha^1_{12}C^2(p^N)^2}{    ((\alpha^1_{11})^2C^2(p^N)^2+\sigma^2)((\alpha^1_{12})^2C^2(p^N)^2+\sigma^2)}
(\alpha^1_{12}+\alpha^1_{11}) + 1 \right],
\end{eqnarray*}
${\dot\alpha}^1_{11}-{\dot\alpha}^1_{12}>0$ if and only if
$\alpha^1_{12}-\alpha^1_{11}>0$.

Fix $\alpha^1_{11}=\alphalbar$.   We need to find out the largest value of $\alpha^1_{12}>\alphalbar$ satisfying ${\dot\alpha}^1_{12}>0$, since we know $\alpha^1_{12}>\alpha^1_{11}=\alphalbar$ implies
${\dot\alpha}^1_{11}>{\dot\alpha}^1_{12}>0$.
A simple calculation shows
\[
  \alpha^1_{12} \le \frac{\sqrt{ \alphalbar C^2 (p^N)^2 -\sigma^2}}{C p^N}
\]
and
\[
  {\dot\alpha}^1_{12}>0.
\]
  By the same token,
\[
  \alpha^1_{11} \le \frac{\sqrt{ \alphalbar C^2 (p^N)^2 -\sigma^2}}{C p^N}=\bigO(\sigma)
\]
if $\alpha^1_{12}=\alphalbar$, and $\alpha^1_{11}\ge \alphalbar$.   Define
\begin{equation}
  {\underline\beta}=\frac{\sqrt{ \alphalbar C^2 (p^N)^2 -\sigma^2}}{C p^N}
  \sim \bigO(\sigma),
\label{eq: betalbar}
\end{equation}
since $\alphalbar=\bigO(\sigma^2)$.
\end{proof}

For fixed $\sigma^2$, choose $(\alphalbar,\betalbar)$ according to Lemma \ref{lm: threshold}.   We are interested in the probability of an event
\[
\exists T\ge 1, \  (\alpha^1_{11,t},\alpha^1_{12,t})\in {\mathcal W}(\alphalbar,\betalbar)
\]
from the neighborhood of the Nash equilibrium outcome $(\alpha^1_{11,t},\alpha^1_{12,t})=(0,0)$  which is the stable stationary outcome.
Since $(0,0)\not\in {\mathcal W}(\alphalbar,\betalbar)$, and $(\epsilon_{1,t},\epsilon_{2,t})$ have a good rate function, there is a real-valued function $I^\alpha: {\mathcal W}(\alphalbar,\betalbar)\rightarrow {\mathbb R}$ such that
for any closed subset $F$ and any open subset $G$ of ${\mathcal W}(\alphalbar,\betalbar)$
\begin{eqnarray*}
&&  -S^*(\alpha^1,G,\sigma^2)\equiv  -\inf_{x\in G}I^\alpha(x) \\
  &\le & \lim_{K\rightarrow 0}\frac{1}{m(t_K+\tau)-t_K}\log\Prob \left( \exists \tau>0, \  \alpha^{1,K}_{1}(\tau)=(\alpha^{1,K}_{11}(\tau),\alpha^{1,K}_{12}(\tau))\in {\mathcal W}(\alphalbar,\betalbar) \ | \ \alpha^1_{1,1}=(0,0)\right) \\
  &\le & -\inf_{x\in F}I^\alpha(x)
\end{eqnarray*}

We now prove Lemma \ref{lemma-pi1}.

\begin{proof}
Let $t_K$ is the time when $(\alpha^1_{1,t_K},\alpha^1_{2,t_K})=(0,0)$.    We are interested in the evolution of $\alpha^1_{i,t}$ over the (fictitious) time interval $\tau$ when
\[
\alpha^1_{i,t}= \frac{\pi_j\sum_{t=t_K}^{m(t_K+\tau)}
    \epsilon_{1t}\epsilon_{2t}}{\sum_{t=t_K}^{m(t_K+\tau)}\epsilon^2_{it}}\ge\alphalbar
\]
occurs.    Since $b^1_{i,t}=p^N$, and $p_{i,t}=b^1_{i,t}+\epsilon_{i,t}$,
\[
  \alpha^{1}_{i,t} =
  \frac{\pi_j\sum_{t=t_K}^{m(t_K+\tau)}\epsilon_{1t}\epsilon_{2t}}{
    \sum_{t=t_K}^{m(t_K+\tau)}}.
\]
Ignoring the difference in period $m(t_K+\tau)$, we write
\[
  \frac{\pi_j\sum_{t=t_K}^{m(t_K+\tau)}
    \epsilon_{1t}\epsilon_{2t}}{\sum_{t=t_K}^{m(t_K+\tau)}\epsilon^2_{it}}=\alphalbar.
\]
At the same time,
\[
  \alphalbar\le
  \alpha^{1,K}_{1j,t}=\frac{\pi_i\sum_{t=t_K}^{m(t_K+\tau)}\epsilon_{1t}\epsilon_{2t}}{
    \sum_{t=t_K}^{m(t_K+\tau)}\epsilon^2_{jt}} \le\betalbar.
\]
After a tedious calculation, one can show that the event
\[
\alpha^{1,K}_i(\tau)\in  {\mathcal W}(\alphalbar,\betalbar)
\]
is equivalent to
\[
  \left\{ 1\le\frac{\pi_1 \sum_{t=t_K}^{m(t_K+\tau)}\epsilon^2_{1,t}}{\pi_2\sum_{t=t_K}^{m(t_K+\tau)}\epsilon^2_{2,t}}
    \le \frac{\betalbar}{\alphalbar} \right\}\cup
  \left\{ 1\le\frac{\pi_2 \sum_{t=1}^T\epsilon^2_{2,t}}{\pi_1\sum_{t=1}^T\epsilon^2_{1,t}} \le \frac{\betalbar}{\alphalbar} \right\}.
\]
Define
\[
{\mathcal L}_i=\left\{ 1\le\frac{\pi_i \sum_{t=t_K}^{m(t_K+\tau)}\epsilon^2_{1,t}}{\pi_j\sum_{t=t_K}^{m(t_K+\tau)}\epsilon^2_{j,t}}
    \le \frac{\betalbar}{\alphalbar} \right\} \qquad\forall i\ne j\in\{1,2\}.
\]
To simplify notation, let us calculate the probability of ${\mathcal L}_1(\alphalbar,\betalbar)$, where at time $\tau$, $\alpha^{1,K}_{1}(\tau)=\alphalbar$ and
$\alphalbar\le\alpha^{1,K}_{2}(\tau)\le \frac{\betalbar}{\alphalbar}$.   The probability of ${\mathcal L}_2(\alphalbar,\betalbar)$ can be calculated by following the same logic.

Note that the event ${\mathcal L}_1(\alphalbar,\betalbar)$
\[
  1\le\frac{\pi_1 \sum_{t=t_K}^{m(t_K+\tau)}\epsilon^2_{1,t}}{\pi_2\sum_{t=t_K}^{m(t_K+\tau)}\epsilon^2_{2,t}}
  \le \frac{\betalbar}{\alphalbar}
\]
is equivalent to
\begin{eqnarray*}
\frac{1}{m(t_K+\tau)}\sum_{t=t_K}^{m(t_K+\tau)}\left( \pi_1\epsilon^2_{1,t} -\pi_2\epsilon^2_{2,t} \right) & \ge & 0 \\
\frac{1}{m(t_K+\tau)}\sum_{t=t_K}^{m(t_K+\tau)}\left( \pi_1\epsilon^2_{1,t} -\frac{\betalbar}{\alphalbar}\pi_2\epsilon^2_{2,t} \right) & \le & 0
\end{eqnarray*}
which is, in turn, equivalent to
\begin{eqnarray*}
\frac{1}{m(t_K+\tau)}\sum_{t=t_K}^{m(t_K+\tau)}\left( \pi_1\frac{\epsilon^2_{1,t}}{\sigma^2_{t_K}} -\pi_2\frac{\epsilon^2_{2,t}}{\sigma^2_{t_K}} \right) & \ge & 0 \\
\frac{1}{m(t_K+\tau)}\sum_{t=t_K}^{m(t_K+\tau)}\left( \pi_1\frac{\epsilon^2_{1,t}}{\sigma^2_{t_K}} -\frac{\betalbar}{\alphalbar}\pi_2\frac{\epsilon^2_{2,t}}{\sigma^2_{t_K}} \right) & \le & 0.
\end{eqnarray*}
Given $\tau>0$, we want to calculate
\begin{equation}
\lim_{K\rightarrow\infty}\Prob \left( {\mathcal W}_1(\alphalbar,\betalbar)\right).
\label{eq: limit probability}
\end{equation}
Since $\forall \tau>0$,
\[
  \lim_{K\rightarrow\infty}\frac{\sigma^2_t}{\sigma^2_{t_K}}=1
  \qquad\forall t\in \{ t_K,\ldots,m(t_K+\tau)\},
\]
\[
\lim_{K\rightarrow\infty}\Prob \left( {\mathcal W}_1(\alphalbar,\betalbar)\right) \\
 = \lim_{K\rightarrow\infty}\Prob \left(
\begin{matrix}
\frac{1}{m(t_K+\tau)-t_K} \sum_{t=t_K}^{m(t_K+\tau)}\left( \pi_1\frac{\epsilon^2_{1,t}}{\sigma^2_t} -\pi_2\frac{\epsilon^2_{2,t}}{\sigma^2_t} \right)    \ge 0 \\[2pt]
\frac{1}{m(t_K+\tau)-t_K}\sum_{t=t_K}^{m(t_K+\tau)}\left( \pi_1\frac{\epsilon^2_{1,t}}{\sigma^2_t} -\frac{\betalbar}{\alphalbar}\pi_2\frac{\epsilon^2_{2,t}}{\sigma^2_t} \right)  \le  0
\end{matrix}
      \right).
\]
Note that $\forall t\in \{t_K,\ldots, m(t_K+\tau)\}$,
$\frac{\epsilon^2_{i,t}}{\sigma^2_t}$ is i.i.d.
Let $x_{i,t}$ be a random variable drawn from the same distribution as $\frac{\epsilon^1_{i,t_K}}{\sigma^2_{t_K}}$ for $t\in \{t_K,\ldots, m(t_K+\tau)\}$.
clearly, $x_{i,t}$ is i.i.d. with mean 1 and a finite variance.   Then, \eqref{eq: limit probability} is equal to
\[
     \lim_{K\rightarrow\infty}\Prob \left(
\begin{matrix}
\frac{1}{m(t_K+\tau)-t_K}\sum_{t=t_K}^{m(t_K+\tau)}\left( \pi_1x_{1,t}-\pi_2x_{2,t} \right)  \ge 0 \\[2pt]
\frac{1}{m(t_K+\tau)-t_K}\sum_{t=t_K}^{m(t_K+\tau)}\left( \pi_1x_{1,t}-\frac{\betalbar}{\alphalbar}\pi_2x_{2,t} \right)  \le  0
\end{matrix}
      \right).
\]
Since $\betalbar\sim\bigO(\sigma_{t_K})$ and $\alphalbar\sim\bigO(\sigma^2_{t_K})$,
\[
\lim_{K\rightarrow\infty}\frac{\betalbar}{\alphalbar}=\infty.
\]
Define
\[
  {\mathcal L}^*_1(\sigma)=\left\{ (x_1,x_2) \ \mid \ \pi_1 x_1 -\pi_2 x_2 \ge 0 \ \text{and} \
  \pi_1 x_1-\frac{\betalbar}{\alphalbar}\pi_2 x_2 \le 0 \right\}
\]
since $\betalbar$ and $\alphalbar$ is parameterized by $\sigma$.  Recall that $(\Expect x_{1,t},\Expect x_{2,t})=(1,1)$.    If $\pi_1 \ge \pi_2$, then
\[
(\Expect x_{1,t},\Expect x_{2,t})\in {\mathcal L}^*_1(\sigma)
\]
for any sufficiently large $K$.   Moreover, since $\pi_1$ and $\pi_2$ are independent of $\sigma^2$,
\[
  \inf_{\sigma^2>0}\lim_{K\rightarrow\infty}\Prob
  \left( {\mathcal L}_1(\alphalbar,\betalbar)\right) >0,
\]
which implies that
\[
-\lim_{\sigma^2\rightarrow 0}\lim_{\tau\rightarrow 0} \lim_{K\rightarrow\infty}\frac{1}{t_K}\log\Prob \left( {\mathcal W}_1(\alphalbar,\betalbar)\right) =0.
\]
Thus, the large deviation rate function must be 0.

If $\pi_1 < \pi_2$, then $\forall\sigma>0$,
\begin{equation}
(\Expect x_{1,t},\Expect x_{2,t})\not\in {\mathcal L}(\sigma).
\label{eq: not here}
\end{equation}
Recall that $\epsilon_{i,t}$ has a finite good rate function, and $x_{i,t}$ is a continuous function of $\epsilon_{i,t}$.   Thus, the contraction principle implies that $x_{i,t}$ has a finite good rate function, say $S^x ( {\mathcal L}^*_1(\sigma))$.   Moreover, if $\sigma'<\sigma$,
\[
  {\mathcal L}(\sigma)\subset   {\mathcal L}(\sigma').
\]
Hence,
\[
S^x ( {\mathcal L}^*_1(\sigma)) \ge S^x({\mathcal L}^*_1(\sigma')).
\]
Therefore,
\[
\limsup_{\sigma\rightarrow 0} S^x({\mathcal L}^*_1(\sigma)) <\infty.
\]
Since $S^x$ is a good rate function, for any open subset $G\subset {\mathcal L}^*_1(\sigma)$,
\[
  \inf_{\sigma^2>0}\lim_{K\rightarrow\infty}\Prob \left( {\mathcal L}_1(\alphalbar,\betalbar)\right) \ge -\inf_{G\subset {\mathcal L}^*_1(\sigma)} S^x(G) >-\infty.
\]
Combining the two cases, we conclude that
\[
-\lim_{\sigma^2\rightarrow 0}\lim_{\tau\rightarrow 0}\lim_{K\rightarrow\infty}\frac{1}{t_K}\log\Prob \left( {\mathcal L}_1(\alphalbar,\betalbar)\right) <\infty.
\]
\end{proof}

\subsection{Dynamics of $\pi_{i,t}$}\label{app: pi}
The following is the proof of Proposition \ref{pr: first main result} in the text.
\begin{proof}
We need to show that $\forall\tau>0$, $\forall\nu>0$, $\exists K(\nu,\tau)$ such that $\forall K>K(\nu,\tau)$,
\[
  \Prob \left( \frac{1}{m(t_K+\tau)-t_K}\sum_{t=t_K}^{m(t_K+\tau)} \Pi^1_{i,t}  -\Pi^N >0    \right) >1-\nu.
\]
Fix $\{\sigma^2_K\}_{K=1}^\infty$ where $\sigma^2_K\rightarrow 0$.
Fix an arbitrarily small $\tau>0$, $\nu'>0$, a decreasing sequence $\{\mu_K\}_{K=1}^\infty$ of positive numbers with
  \[
\lim_{K\rightarrow\infty}\mu_K=0
  \]
  to examine successively smaller neighborhood ${\mathcal N}_{\mu_K}(0,0)$ of $(\alpha^1_1,\alpha^1_2)=(0,0)$ but
  \[
    (\alphalbar,\alphalbar)\in{\mathcal N}_{\mu_K}(0,0) \qquad\forall K.
  \]
Let $S^*_i(\mu_K,\sigma^2_K)$ be the rate function that $(\alpha^1_{11,t},\alpha_{12,t})$ escapes from ${\mathcal N}_{\mu_K}(0,0)$ starting from the stable stationary point $(0,0)$ of the mean dynamics of $(\alpha^1_{11,t},\alpha_{12,t})$.   By the definition of the rate function, we know that $S^*(\mu_K,\sigma^2_K)$ is an increasing function of $\mu_K$ with
  \[
    \lim_{\mu_K\rightarrow 0}S^*(\mu_K,\sigma^2_K)=0.
  \]
  and for a fixed $\mu_K$,
  \[
    \limsup_{\sigma^2_K\rightarrow 0}S^*(\mu_K,\sigma^2_K) <\infty.
  \]
  Therefore,
  \[
S^*(\mu_K,\sigma^2_K)\downarrow 0
  \]
as $K\rightarrow\infty$.   For each $K\in\{1,2,\ldots\}$, choose $t_K$ according to
  \[
    t_K =\left\lceil -\frac{\log\nu'}{S^*(\mu_k,\sigma^2_K)} \right\rceil.
  \]
so that
  \[
e^{-t_K S^*(\mu_k)}=\nu'
  \]
ignoring the integer problem. As $K\rightarrow\infty$, $t_K\rightarrow\infty$.
Since ${\lambda_t}=1/t$, $m(t_K+\tau)-t_K\rightarrow\infty$ as $K\rightarrow\infty$.

Let us consider random variable $\Pi^1_{i,t}$.   The probability that $(\alpha^1_{11,t},\alpha^1_{12,t})$ exists ${\mathcal N}_{\mu_K}$ between $t_K$ and $m(t_K+\tau)$ converges to $e^{-t_K S^*(\mu_K,\sigma^2_K)}$ (\citet*{DupuisandKushner89}).  Upon exiting ${\mathcal N}_{\mu_K}(0,0)$, the mean dynamics of $(\alpha^1_{11,t},\alpha^1_{12,t})$ is pointing away from $(0,0)$ but to $(1,1)$ so that $\Pi^1_{i,t}$ strictly increases.   Thus, with probability $e^{-t_K S^*(\mu_K,\sigma^2_K)}$, $\Pi^1_{i,t}=\Pi^N+ \zeta(\tau)$ where $\zeta(\tau)>0$ with $\zeta(0)=0$.
Thus, $\forall t\in \{t_K,\ldots,m(t_K+\tau)\}$,
\[
  \Expect \Pi^1_{i,t}-\Pi^N \ge e^{-t_K S^*(\mu_K,\sigma^2_K)}(\Pi^N+\zeta(\tau)-\Pi^N)
  -\bigO(\mu_K) \left( 1-e^{-t_K S^*(\mu_K,\sigma^2_K)}\right)
\]
under the assumption that within ${\mathcal N}_{\mu_K}$, $\Pi^1_{i,t} \le\Pi^N$, but the difference must be bounded by $\bigO(\mu_K)$.

Since $e^{-t_K S^*(\mu_K,\sigma^2_K)}=\nu'$, and $\bigO(\mu_K)\rightarrow 0$ as $K\rightarrow 0$, $\exists K'$ such that $\forall K\ge K'$, $\forall t\in \{t_K,\ldots,m(t_K+\tau)\}$,
\[
  \Expect \Pi^1_{i,t} -\Pi^N \ge \frac{1}{2}\nu'\zeta(\tau).
\]
Thus, $\forall \nu\in \left( 0, \frac{1}{4}\nu'\zeta(\tau)\right)$, $\exists K(\nu,\tau)$ such that
\[
  \Prob\left(
    \frac{1}{m(t_K+\tau)-t_K}\sum_{t=t_K}^{m(t_K+\tau)}\Pi^1_{i,t}-\Pi^N >\nu
  \right) > 1-\nu
\]
from which the conclusion follows.
\end{proof}

\section{Convergence and Stability Calculations}
\subsection{Convergence and Stability under $\mathcal{M}^0$}\label{app:M0-conv}

In the specification $\mathcal{M}^0$, each firm $i$ forecasts its opponent's price $p_{j,t}$ by forming a discounted average of past prices. We denote the current estimate $\alpha^0_{i,t}$, and the best response function which determines the expected price is
\[
  b^0(\alpha^0_i)= \frac{A+C\alpha^0_i}{2B}.
\]
Everything is symmetric in the model, so the best response functions are the same, only the beliefs differ. We stack the beliefs into the vector $\alpha^0_t=[\alpha^0_{1,t},\alpha^0_{2,t}]'$ whose evolution follows:
\begin{eqnarray}
  \left[
    \begin{matrix}
      \alpha^0_{1,t}\\
      \alpha^0_{2,t}
      \end{matrix}
    \right] & = &
      \left[
    \begin{matrix}
      \alpha^0_{1,t-1}\\
      \alpha^0_{2,t-1}
      \end{matrix}
    \right] +\lambda_t
      \left[
    \begin{matrix}
      p_{2,t}-\alpha^0_{1,t-1} \\
      p_{1,t}-\alpha^0_{2,t-1}
      \end{matrix}
    \right]  \nonumber \\
    & = &
      \left[
    \begin{matrix}
      \alpha^0_{1,t-1}\\
      \alpha^0_{2,t-1}
      \end{matrix}
    \right] +\lambda_t
      \left[
    \begin{matrix}
      b^0(\alpha^0_{2,t-1})-\alpha^0_{1,t-1} + \epsilon_{2,t} \\
      b^0(\alpha^0_{1,t-1})-\alpha^0_{2,t-1} +  \epsilon_{1,t}
      \end{matrix}
    \right]  \label{S0-update}
\end{eqnarray}
This last equation says that the increment in beliefs $\alpha^0_{t}-\alpha^0_{t-1}$ is equal to the gain $\lambda_t$ multiplied by a stochastic update term.  We can then think of $\lambda_t$ as the notional time between observations and analyze the limit as $\lambda_t  \rightarrow 0$. Along this limit the increment between observations converges to the time derivative, and since the impact of the shocks $\epsilon_{i,t}$ is of order $\lambda$ they become asymptotically negligible. Thus as $\lambda_t \rightarrow 0$ the stochastic belief process
$\{\alpha^0_t\}$ converges to the stable stationary solution of the ordinary differential equation:
\[
  \dot \alpha^0 (t) = g(\alpha^0(t)) \equiv \left[\begin{matrix}
      b^0(\alpha^0_{2}(t))-\alpha^0_{1}(t)  \\
      b^0(\alpha^0_{1}(t))-\alpha^0_{2}(t)
      \end{matrix}
    \right]
\]
This linear differential equation has a unique equilibrium point $\bar \alpha^0$ which is symmetric and where beliefs are the Nash equilibrium prices:
\[  \dot \alpha^0 (t) =  0 \Rightarrow  \left[\begin{matrix}
      b^0(\bar \alpha^0_{2})-\bar \alpha^0_{1}  \\
      b^0(\bar \alpha^0_{1})-\bar \alpha^0_{2}
      \end{matrix}
    \right] =0
    \Rightarrow \bar \alpha^0_{1} = \bar \alpha^0_{2} \Rightarrow  \bar \alpha^0_i  = b(\bar \alpha^0_i) = \frac{A}{2B-C} = p^N.
\]
Moreover, since the differential equation is linear, the global stability of the equilibrium is determined by the eigenvalues of the Jacobian matrix:
\[ g'(\bar \alpha) = \left[\begin{array}{rr}
      -1 & \frac{C}{2B}  \\
      \frac{C}{2B} & -1
      \end{array}\right], \]
which has eigenvalues equal to $-1 \pm \frac{C}{2B}$. Since we have already assumed that $C<B$, both eigenvalues are negative, and the vector $\bar a$ is globally stable equilibrium point of the ODE.
\section{Convergence and Stability under $\mathcal{M}^1$}\label{app:M1-conv}

The analysis of beliefs under specification $\mathcal{M}^1$ is conceptually very similar, but now the dimensionality of the problem grows. Moreover, the analysis becomes slightly more complex because the dependence of the key equations on beliefs and shocks is nonlinear.   In particular, under $\mathcal{M}^1$ each firm estimates
a vector $\alpha^1_{i,t}=[\alpha^1_{0i,t}, \alpha^1_{1i,t}]'$ of the intercept and slope parameters in the reaction function \eqref{eq: perceived law 1}. Furthermore, in the recursive least squares algorithm \eqref{eq: recursive least square}, each firm also estimates the second moment matrix $R_{i,t}$.
This matrix is symmetric, and so we can economize by only tracking only the two distinct time-varying elements, which we collect in a vector $r_{it} = [R_{i,t}(1,2), R_{i,t}(2,2) ]'$.  Thus for each firm, we can stack the
beliefs into the $4\times 1$ column vector
\[
  \theta_{i,t} = [(\alpha^1_{i,t})',r_{i,t}']',
\]
and similarly can stack the beliefs of the two firms as
\[
  \alpha^1_t=[(\alpha^1_{1,t})',(\alpha^1_{2,t})']'
\]
and
\[
  \theta_t = [\theta_{1,t}',\theta_{2,t}']'.
\]
As above, we can write the best response of a firm given its current beliefs $\alpha^1_{i}$ as:
\[
  b(\alpha^1_{i}) = \frac{A+C \alpha^1_{0i}}{2(B-C\alpha^1_{1i})},
\]
which is now nonlinear in $\alpha^1_i$.

Using the best response functions the price realizations, we can rewrite the dynamics of each firm's beliefs from (\ref{eq: recursive least square}) as:
\begin{eqnarray}
\alpha^1_{i,t}
& = & \alpha^1_{i,t-1} +\lambda R^{-1}_{i,t-1}
      \left[
    \begin{matrix}
       b(\alpha^1_{j,t-1}) -\alpha^1_{0i,t-1} + \epsilon_{j,t} -\alpha^1_{1i,t-1}(b(\alpha^1_{i,t-1})+\epsilon_{i,t})  \\
       \left[b(\alpha^1_{j,t-1}) -\alpha^1_{0i,t-1} + \epsilon_{j,t} -\alpha^1_{1i,t-1}(b(\alpha^1_{i,t-1})+\epsilon_{i,t})\right](b(\alpha^1_{i,t-1})+\epsilon_{i,t})
      \end{matrix}
    \right] \nonumber
  \\
  R_{i,t}& = & R_{i,t-1}+\lambda \left( \left[ \begin{matrix}
1   &   b(\alpha^1_{i,t-1})+\epsilon_{i,t} \\
b(\alpha^1_{i,t-1})+\epsilon_{i,t} &  (b(\alpha^1_{i,t-1})+\epsilon_{i,t})^2
\end{matrix}
    \right] -R_{i,t-1} \right). \label{S1-update}
\end{eqnarray}
In addition to being nonlinear, the belief dynamics are now a quadratic function of the shocks.

For the calculations below, it will be convenient to define a vector function capturing the update terms in the $\alpha^1_{i,t}$ dynamics:
\[ G_i(\alpha^1,\epsilon_{i,t},\epsilon_{j,t}) =
\left[
    \begin{matrix}
      G_{0i}(\alpha^1) + \epsilon_{j,t} -\alpha^1_{1i}\epsilon_{i,t}  \\
      G_{0i}(\alpha^1) b(\alpha^1_i) + \left[G_{0i}(\alpha^1)-\alpha^1_{1i}b(\alpha^1_i)\right] \epsilon_{i,t} + b(\alpha^1_i)\epsilon_{j,t} + \epsilon_{i,t}\epsilon_{j,t} -\alpha^1_{1i}\epsilon_{i,t}^2
      \end{matrix}
    \right]
\]
where
\[
  G_{0i}(\alpha^1) = b(\alpha^1_{j}) -\alpha^1_{0i} -\alpha^1_{1i}b(\alpha^1_{i}).
\]
Then the first equation in the belief updating can be written:
\[
  \alpha^1_{i,t} =  \alpha^1_{i,t-1} +\lambda R^{-1}_{i,t-1}  G_i(\alpha^1_t,\epsilon_{i,t},\epsilon_{j,t}).
\]
Everything is symmetric for firm $j$, appropriately switching the roles of own and opponent's beliefs and shocks.   Using the independence of the shocks and that $E \epsilon_{it}^2 = \sigma^2$, we can then define the expected update function as:
\[ g_i(\alpha^1) = E[ G_i(\alpha^1,\epsilon_{i,t},\epsilon_{j,t})] =
\left[
    \begin{matrix}
      G_{0i}(\alpha^1)   \\
      G_{0i}(\alpha^1) b(\alpha^1_i) -\alpha^1_{1i}\sigma^2
      \end{matrix}
    \right]
\]
We can then use the same logic as above to study the small gain limit as $\lambda_t \rightarrow 0$. Relying on the same general stochastic approximation results, as the gain converges to zero the firms' beliefs converge weakly to the solution of the differential equations:
\begin{eqnarray*}
\dot \alpha^1_i(t)  &=& R_i(t)^{-1} g_i(\alpha^1(t)) \\
\dot R_i(t) &=& \left[ \begin{matrix}
1   &   b(\alpha^1_{i}(t) ) \\
b(\alpha^1_{i}(t)) &  b(\alpha^1_{i}(t))^2 + \sigma^2
\end{matrix}
    \right] -R_{i}(t).
\end{eqnarray*}
It is straightforward to see that the equilibrium of this specification leads to the same outcomes as the case of ${\mathcal M}^0$. In particular, since $R_i(t)$ is nonsingular, an equilibrium point $\bar \alpha^1$ satisfies:
\begin{equation}\label{eq cond}
\dot \alpha^1_i = 0 \Rightarrow g_i(\bar \alpha^1) = 0 \Rightarrow
      G_{0i}(\bar \alpha^1) =0  \text{ and } G_{0i}(\bar \alpha^1) b(\bar \alpha^1_i) -\bar \alpha^1_{1i}\sigma^2 = 0.
\end{equation}
The only way to satisfy both of the last two conditions is to have $\alpha^1_{1i}=0$.
Since this holds for $i=1,2$, we can then reduce $G_{0i}(\bar \alpha^1) = b(\bar \alpha^1_{j}) -\bar \alpha^1_{0i}$, so $[G_{01}(\bar \alpha^1),G_{02}(\bar \alpha^1)]=g([\bar \alpha^1_{01},\alpha^1_{02}]')$, the same function as in $\mathcal{M}^0$ above that defines
an equilibrium in that specification. Thus the
unique equilibrium is symmetric and satisfies $\bar \alpha^1_i = [p^N,0]'$, $i=1,2$.  Given $\bar \alpha^1$, the second moment matrices $R_i(t)$ converge
to the same limit point:
\[ \bar R = \left[ \begin{matrix}
1   &   p^N \\
p^N &  (p^N)^2 + \sigma^2
\end{matrix}
    \right].
\]
Verifying the local stability\footnote{since differential equation is nonlinear} is straightforward, if tedious. In particular, for local analysis we can decouple the equations in $\alpha^1$ and $R$. The stability of $\bar R$ is immediate, and while the local stability of $\bar \alpha^1$ depends on the eigenvalues of the Jacobian matrix of $g(\alpha^1)$. Via simple calculations we see that this block symmetric matrix can be written:
\[ g_\alpha (\bar \alpha^1) = \left[ \begin{matrix}
-1 & -2p^N & \frac{C}{2B} & 2 p^N \\
-p^N & -2(p^N)^2 - \sigma^2 & p^N \frac{C}{2B} & 2 (p^N)^2 \\
\frac{C}{2B} & 2p^N & -1 & -2p^N \\
p^N\frac{C}{2B} & 2(p^N)^2 & -p^N & -2(p^N)^2 - \sigma^2 \end{matrix} \right].
\]
Simple calculations show that two of the eigenvalues of this matrix are given by $-1 + \frac{C}{2B}$ and $-\sigma^2$ with the other two being:
\[
 -\frac{1}{2}\left(1+ \frac{C}{2B} + 4 (p^N)^2 + \sigma^2\right)\pm \frac{1}{2}\sqrt{\left(1+ \frac{C}{2B} + 4 (p^N)^2 + \sigma^2\right)^2-4\left(1+ \frac{C}{2B}\right)\sigma^2}.
\]
Clearly all are negative.\footnote{We use the result that for a block symmetric matrix $\left[\begin{matrix} A & B \\ B & A \end{matrix}\right]$ its eigenvalues are
the union of the eigenvalues of $A+B$ and $A-B$.}

While the equilibrium, which supports that Nash equilibrium outcomes, is stable, the mean dynamics which draw belief toward the equilibrium are weak
along some directions. For example, two of the eigenvalues of the Jacobian matrix are very small, and vanish as $\sigma^2 \rightarrow 0$.  This is a reflection of the singularity of $g_\alpha (\bar \alpha^1)$ for $\sigma^2=0$, which in turn is evident in the equilibrium conditions (\ref{eq cond}). As $\sigma^2 \rightarrow 0$ the two equilibrium restrictions in the last expression in (\ref{eq cond}) reduce to only $G_{0i}(\alpha^1)=0$.

\begin{figure}
\centerline{\includegraphics[width=0.8 \textwidth]{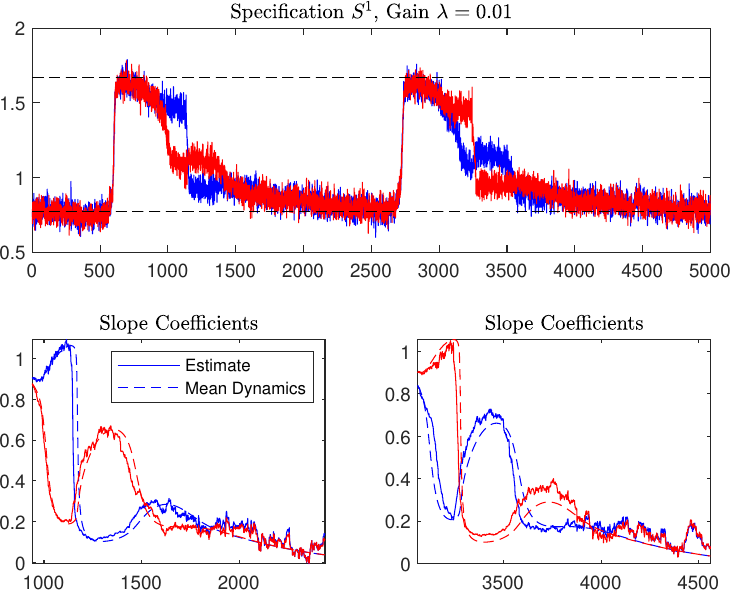}}
\caption{Simulated prices and beliefs in the duopoly model with specification $\mathcal{M}^1$. Top panel: simulated prices. Bottom panels: estimated slope coefficients and mean dynamics.}\label{fig-meand1}
\end{figure}

We illustrate the mean dynamics and convergence in Figure \ref{fig-meand1}. In the simulations above we saw asymmetries in the price cutting episodes, which lead firms to eventually cut their prices from near the joint monopoly price $p^C$ to the Nash equilibrium price $p^N$. The top panel of Figure \ref{fig-meand1} shows a similar time series of prices. The escapes from the equilibrium appear symmetric but the convergence process is asymmetric. In the first price-cutting round, starting around period 1,000 in the figure, firm 2 (in red) initially undercuts firm 1 (in blue), who later drops prices further.  In the second episode, starting around period 3,100, the roles are reversed. These differences in behavior are driven by small initial differences in beliefs, which in turn are caused by specific shock realizations. The bottom panels of the figure show the dynamics of the simulated reaction function slope coefficients for the two firms in each of the two episodes, along with the solution of the mean dynamics differential equations initialized at particular points in the simulations. The dynamics of the stochastic simulations closely follow the deterministic path of the mean dynamics. Small asymmetries in the initial conditions, driven by the small differences in the realizations of shocks after the escapes, lead to very marked differences in the dynamics of beliefs and actions.

\section{Large Deviations for a Fixed Specification}

\subsection{General Large Deviation Results}\label{app: LD background}

Here we present some general large deviation results for a fixed belief specification, drawn from our earlier work \citet*{ChoWilliamsandSargent02} and \citet*{Williams2019}.
We start by writing the belief updating equations in the general form:
\[
  \theta_{t} = \theta_{t-1} + \lambda \Psi(\theta_t,\epsilon_t),
\]
where $\theta$ is the composite belief vector for the two firms and  $\epsilon=[\epsilon_1,\epsilon_2]$ the composite shock vector.   In specification $\mathcal{M}^0$, $\theta_t = \alpha^0_t$ and  (\ref{S0-update}) shows that $\psi$ is linear in $\epsilon_t$.  In specification $\mathcal{M}^1$, we have already seen that $\theta_t$ is the stacked vector summarizing $(\alpha^1_{it},R_{it})$ for the two firms, and (\ref{S1-update}) shows that $\psi$ is a quadratic function of $\epsilon_t$. Then we can write the mean dynamics ODEs governing convergence as:
\[
  \dot \theta(t) = \psi(\theta(t) \equiv E[\Psi(\theta(t),\epsilon_t)] .
\]
We  develop a deterministic control problem whose solution determines the frequency of large deviations as well the most likely path beliefs follow when they exit a region containing the equilibrium point.  This control problem represents the beliefs as following a perturbed version of the mean dynamics:
\begin{equation}
\dot \theta = \psi(\theta) + v .\label{perturbed ODE}
\end{equation}
The ``cost'' of a particular perturbation $v$ depends on its  size relative to the volatility of beliefs, as bigger perturbations will more naturally occur with more volatility.  If the beliefs were normally distributed, then their variance would be a natural measure of volatility.  However since the belief dynamics are quadratic, we need a measure of its volatility which appropriately captures the tail behavior of beliefs.  For a general i.i.d.\ random variable, we can summarize its distribution by its moment generating function, or equivalently the log of it, which is also known as the cumulant generating function. Thus for the i.i.d\ random vector $\epsilon_{t}$, given beliefs $\theta \in \mathbb{R}^{n_\theta}$ and a fixed vector  $\beta \in \mathbb{R}^{n_\theta}$,
the log moment generating function of the beliefs (or belief increments) as:
\begin{eqnarray}
H(\theta,\beta)&=&\log E \exp \left\langle \beta ,\Psi(\theta,\epsilon_t)\right\rangle.  \label{H-static} \\
&=&  \left\langle \beta ,\psi(\theta)\right\rangle  + \log E \exp \left\langle \beta ,v_t \right\rangle, \nonumber
\end{eqnarray}\unskip
where the second line defines the perturbation as
\[
  v_t=\Psi(\theta,\epsilon_t)-\psi(\theta).
\]
We explicitly evaluate this function in the two specifications below.   Then we define the Legendre transform of the $H$  function as:
\begin{eqnarray}
L (\theta,v)&=&\sup_{\beta}
 \left[ \langle \beta, \psi(\theta) + v \rangle
  -H (\theta,\beta) \right] \label{LeGendre} \\
  &=& \sup_{\beta}
 \left[ \langle \beta,  v \rangle -  \log E \exp \left\langle \beta ,v \right\rangle \right] \nonumber.
\end{eqnarray}\unskip
The function $L$ plays the role of the instantaneous cost function for the belief perturbations $v$. As \citet*{DemboandZeitouni98} emphasize, it captures the tail behavior of the distribution of $\Psi(\theta,\epsilon)$, which is crucial for analyzing the large deviations. We calculate $L$ separately in the two specifications.

We now analyze the rate at which the beliefs escape from some set containing the equilibrium $\bar \theta$ in some given finite horizon.   Thus we fix a set $G$ with $\bar \theta \in G$, and $\bar T < \infty$.   We then choose the perturbations $v$
in \eqref{perturbed ODE} which push beliefs to the boundary $\partial G$ in the most cost effective way, where the instantaneous cost is measured by the $L$ function.  Thus we solve the control problem:
\begin{equation} \label{Scrit}
  \bar{S} = \inf_{v(\cdot), T}  \int_{0}^T
  L (\theta(s),v(s))  ds
\end{equation}
where the minimization is subject to the dynamics
(\ref{perturbed ODE})  and:
\begin{eqnarray}
 \theta(0) &=&\bar{\theta},\text{ }\theta(T)\in \partial G%
               \text{ for some }0<T\le\bar{T}.  \label{bounds}
 \end{eqnarray}\unskip
If $v \equiv 0$ then the beliefs follow the mean dynamics. The cost is zero, but the beliefs do not escape. To find the most probable escape path, we find a least cost path of
perturbations that pushes beliefs from $\bar{\theta}$ to the boundary of $G$.

To emphasize the dependence of the stochastic beliefs on the gain $\lambda$, we write the belief process as $\{ \theta_t^\lambda\}$. We now define the set of escape paths $\Gamma^\lambda (G,\bar t)$ as the set of all beliefs sequences that start sufficiently close to the equilibrium, $\|\theta_0^\lambda-\bar \theta\|<\delta$ and then exit the set $G$ at some finite time $\tau$, $\theta_\tau \notin G$ for $\tau \le \bar t = \bar T/\lambda$.  Note that by fixing the continuous time horizon $\bar T$, the discrete horizon $\bar t$ increases as the gain shrinks. Finally, let $\theta^\lambda(t)$ be the
continuous interpolation of $\{x_t^\lambda\}$, so an escape will exit the set $G$ at some time $\tau^\lambda$ with $x^\lambda(\tau^\lambda) \in \partial G$.

Then the following theorem, adapted from \citet*{Williams2019} and \citet*{KushnerandYin03} provides the large deviation principle characterizing the escapes. A similar result was stated in \citet*{ChoWilliamsandSargent02}.

\begin{theorem}\label{thm-ld}
Let $\theta^\lambda(\cdot)$ be the piecewise linear interpolation of $\{\theta_t \}$ for fixed gain $\lambda$, and let
\[
  \theta(\cdot):[0,\bar{T}] \rightarrow \mathbb{R}^{n_\theta}
\]
solve \eqref{Scrit}. Then, for any $\theta^{\lambda }(\tau ^{\lambda })$ and $\delta > 0,$ we have the following:
\begin{eqnarray*}
\limsup_{\lambda \rightarrow 0}\varepsilon \log P \left(\theta^{\lambda }(t)\notin G \ \text{ for some } \ 0<t\leq \bar{T}|\theta^{\varepsilon }(0)=\bar x \right) &=& -\bar{S},
\\
 \lim_{\lambda \rightarrow 0} \lambda
\log E(\tau ^{\lambda }) &=& \bar S, \\
\text{ } \lim_{\lambda \rightarrow 0} P \left(  \left. \| \theta^{\lambda }(\tau ^{\lambda })-\theta(T) \| <\delta \right|
\{\theta_t^\lambda\} \in \Gamma^\lambda(G,\bar t) \right) &=& 1.
\end{eqnarray*}
\end{theorem}
The first result shows that the probability of observing an escape on a bounded time interval is exponentially decreasing in the gain $\lambda$, with the rate given by the minimized cost. The second line shows that for small $\lambda$ the log mean  escape times from the becomes close to $\exp(\bar{S})/\lambda)$. The last result shows that with
probability approaching one, escapes occur at the point determined from the solution of the cost minimization problem.

Following standard Hamiltonian-type methods, and using the convex duality
of $H$ and $L$, the solution of the cost minimization problem satisfies the differential equations:
\begin{eqnarray}
\dot{\theta} &=& H_\beta(\theta,\beta) \label{ODE} \\
\dot{\beta}&=& -H_\theta(\theta,\beta) \nonumber
\end{eqnarray}\unskip
subject to the boundary conditions (\ref{bounds}).

\subsection{Large Deviation Calculations for $\mathcal{M}^0$}
\label{app: LD M0}

In specification $\mathcal{M}^0$, calculation of $H$ and $L$ is immediate. Since $\theta=\alpha^0$, and $\Psi$ is linear, $\psi=g(\theta)$, and $v_t=\epsilon_t$, we can write:
\begin{eqnarray*}
H(\theta,\beta) &=&  \left\langle \beta ,g(\theta)\right\rangle  + \log E \exp \left\langle \beta ,\epsilon_t \right\rangle  \\
&=& \left\langle \beta ,g(\theta)\right\rangle  + \frac{\sigma^2}{2} \beta'\beta.
\end{eqnarray*}
where we use that the distributions of $\epsilon_{it}$ are identical, and the shock vector $\epsilon_t$ is normally distributed  with mean zero and variance $\sigma^2 I$. Then we can evaluate:
\[
  L (\theta,v) = \frac{1}{2\sigma^2} v' v.
\]
Thus cost function $L$ weights the perturbations $v$ by the shock variance, just as we discussed above.

The large deviation results for specification $\mathcal{M}^0$ can be calculated explicitly and analytically. First, we modify the consider a slightly modified version of the cost minimization problem (\ref{Scrit}) where we change the terminal boundary condition in (\ref{bounds}) to $\theta(T)=x$.
This defines the least cost path from $\bar \theta$ to $x$, with cost $S(x)$,
and therefore $\bar S = \min_{x \in \partial G} S(x)$.  Via standard methods,
as in \citet*{DemboandZeitouni98}, in the limit as $\bar T \rightarrow \infty$
this function satisfies the HJB equation:
\[
 0 = \min_v \left\{ -\langle S'(x),g(x)+v \rangle  +   \frac{1}{2\sigma^2} v' v \right\} .
\]
The optimal choice of $v={\sigma^2} S'(x)$, when substituted
back into the HJB equation then implies:
\[ S'(\theta) = -\frac{2}{\sigma^2} g(\theta). \]
In turn, this implies that the least-cost escape path  can be found by initializing at $\theta(T) = x$ and then integrating backward in time the following ODE:
\[
  \dot \theta = g(\theta) + v =  g(\theta) +{\sigma^2} S'(\theta) = -g(\theta).
\]
In this case, to find the least cost escape from $\bar \theta$ to $x$, we simply reverse the mean dynamics and follow the convergent path from $x$ to $\bar \theta$.  We also see that by letting $\beta=S'(\theta)$, this solution satisfies the Hamiltonian ODE system \eqref{ODE} above:
\begin{eqnarray*}
\dot{\theta} &=& g(\theta) + \sigma^2 \beta =  g(\theta) - \sigma^2  \frac{2}{\sigma^2} g(\theta) = - g(\theta) \\
\dot{\beta} &=& -\langle \beta , g'(\theta) \rangle =  -\frac{2}{\sigma^2} \langle g'(\theta), \dot{\theta} \rangle = \frac{d}{dt} S'(\theta(t) ).
\end{eqnarray*}
Finally, given that $g(\theta)$ is linear, from the solution for $S'(\theta)$ we can solve explicitly for $S(\theta)=S(a_1,a_2)$:
\[ S(a_{01},a_{02}) = -\frac{2}{\sigma^2} \left( \frac{A}{2B} (a_{01} + a_{02}) + \frac{C}{2B} a_{01} a_{02} - \frac{1}{2}(a_{01}^2 + a_{02}^2) - \frac{A^2}{2B(2B-C)} \right) ,\]
where the constant ensures that $S(\bar a) =0$.

\subsection{Large Deviation Calculations for $\mathcal{M}^1$}\label{app: LD M1}

In this section, we analyze the large deviation properties of the specification $\mathcal{M}^1$.  While under $\mathcal{M}^0$ we could calculate the solution analytically, here we must solve the cost minimization problem (\ref{Scrit}) numerically.

First, we calculate the key elements of the cost minimization problem.   Since in the case of $\mathcal{M}^1$, the belief vector $\theta$ is a larger $8 \times 1$ vector, and the belief update function $\Psi$ is a quadratic form, the calculations are slightly more involved. First, the key terms in calculation of $H$ can be written:
\[
  \langle \beta, \Psi(\theta, \epsilon_t) \rangle =
  V_{00} - V_{01}\epsilon_{t}  -\frac{1}{2}\epsilon_{t}'V_{11}\epsilon_{t}
\]
for a constant $V_{00}$, vector $V_{01}$, and symmetric matrix $V_{11}$.  All of the $V_{ij}$ terms depend on the beliefs $\theta$ as well as the vector $\beta$. As in \citet*{Williams2019}, we can use this expression to evaluate:
\begin{equation}
 H(\theta ,\beta)= V_{00}  -\frac{1}{2}\log|V_{11}+I| + \frac{1}{2}V_{01}(V_{11}+I)^{-1}V_{01}'   \label{H-s1}
\end{equation}
The $H$ function in this case given by a relatively simple expression depending on the $V_{ii}$ matrices, but this can imply some rather complex dependence on the underlying beliefs $\theta$.

Then we analyze escapes from sets defined as closed balls of radius $\rho$ in the belief space:
\[
  G(\rho) = \{ \theta \ | \ \|\alpha^1-\bar \alpha^1\| \le \rho \}.
\]
We focus on the dynamics of the reaction function parameters $\alpha^1_{it}$, rather than the second moment matrix $R_{it}$ whose components $r_{it}$ are also in the $\theta$ vector. For a fixed $G(\rho)$, to solve for $\bar S(\rho)$ from (\ref{Scrit}), we initialize the full parameter vector at the equilibrium $\theta(0)= \bar \theta$ choose an initial vector $\beta(0)=\beta_0$. Then, using the explicit form for $H$ in \eqref{H-s1}, we solve the ODEs \eqref{ODE} from the Hamiltonian system. This requires calculation of the derivatives of $H$, which can be evaluated from \eqref{H-s1}. Note that in the process of finding $\dot \theta$ we find $v(t)= \psi(\theta(t))-\dot \theta(t)$. We supplement the ODE system with another differential equation giving that cumulates the cost :
\[ \dot S(t) = L(\theta(t),v(t)) \]
with $S(0)=0$.
We stop the solution of the ODEs when the beliefs hit the boundary of the set: $\theta(T)=\partial G(\rho)$. This provides a particular exit path, whose cost  $S(T;\beta_0)$ is dependent on the initial condition $\beta_0$. We find the minimum cost path by minimizing over $\beta_0$, which also determines the minimized cost $\bar S$.

\subsection{Relation Between Mean Dynamics and Large Deviations}\label{app: md-ldp}

\begin{figure}
\centerline{\includegraphics[width=0.9 \textwidth]{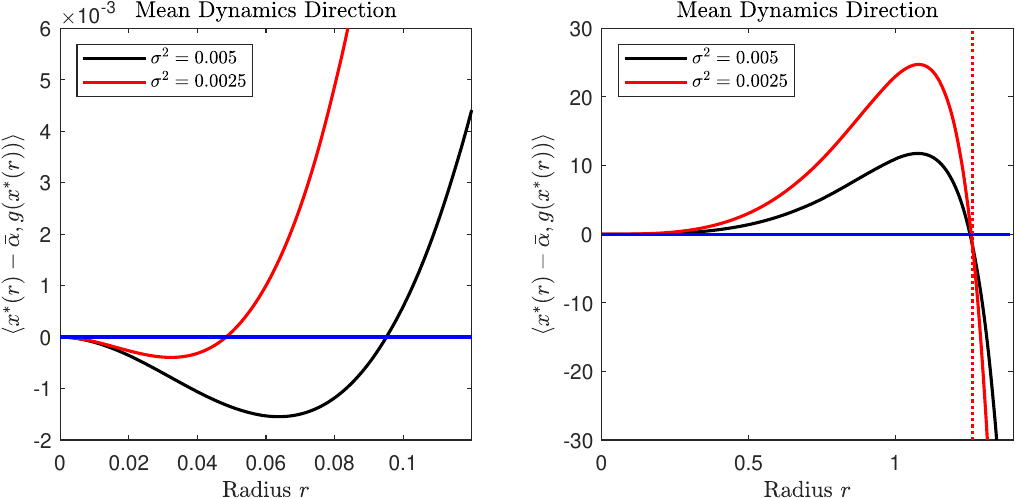}}
\caption{The direction of the mean dynamics $g(x^*(r))$ along a parameterized curve $x^*(r)$ with radius $r$ starting at the stable equilibrium.}\label{fig-escapes12}
\end{figure}

In the text we discussed the relation between the mean dynamics and escape dynamics.  Figure \ref{fig-escapes12} illustrates the directions of the mean dynamics  for two different levels of the shock standard deviation $\sigma$. The model has a unique stable equilibrium, but as $\sigma^2 \rightarrow 0$ there is a continuum of equilibria lying along a surface. Figure  \ref{fig-escapes12} shows that for any positive $\sigma^2 >0$, the equilibrium $\bar \alpha =(0,p^N)$ is locally stable, but the radius of stability shrinks with the noise. In particular, the figure shows the ``direction'' of the mean dynamics $\langle x^*(r)-\bar \alpha, g(x^*(r))\rangle$ along a line $x^*(r)$ connecting the equilibrium $\bar \alpha=(p^N,0)$ and the beliefs $(0,1)$ that support the collusive price $p^C$. When this direction is negative, the mean dynamics point toward the equilibrium, and when it is positive, the mean dynamics point away from the equilibrium.

The left panel focuses on small radii $r$, and the radius of stability of the equilibrium $\bar \alpha$ is given by the intersection with the zero line. The radius of stability is small, and shrinks with the noise $\sigma^2$. The right panel shows the direction of
the mean dynamics over a larger region. Outside the small radius of stability (which is hardly visible in this panel), the mean dynamics point away from the equilibrium for quite some distance until finally turning back. The upper intersection with zero, where the direction hits zero from above, gives a radius $r^e$, which in turn fixes the location of the belief vector $\alpha^e=x^*(r^e)$. With more volatile shocks, the mean dynamics stop short of the beliefs $(0,1)$ supporting the joint monopoly price, whose radius is shown with a vertical dotted line. But as $\sigma^2 \rightarrow 0$ the $\alpha^e$ converges to $(0,1)$.

\begin{figure}
\centerline{\includegraphics[width=0.6 \textwidth]{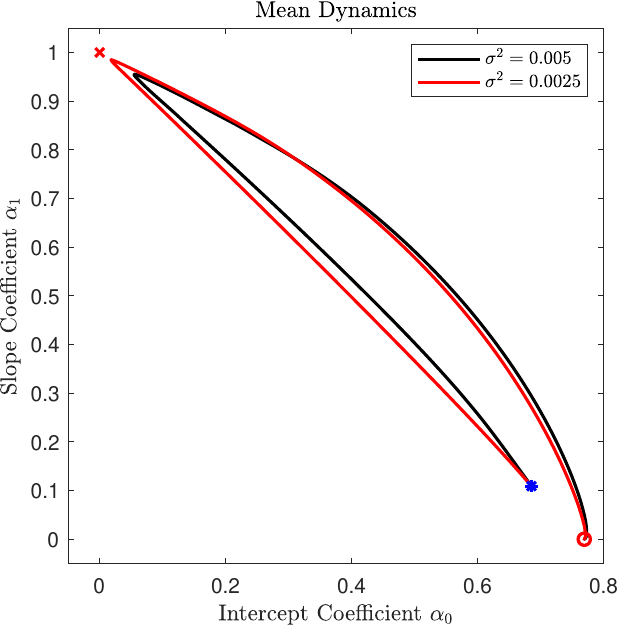}}
\caption{Mean dynamics $g(\alpha)$ from a point near the equilibrium
but outside the stability radius  for two different standard deviations $\sigma$.}\label{fig-escapes13}
\end{figure}

The mean dynamics are illustrated in Figure \ref{fig-escapes13} for the same two levels of the shock standard deviation $\sigma$. The figure shows the solution of the mean dynamics ODE initialized at a point is just outside the stability radius. The initial condition of the ODE is shown with a blue asterisk, while the equilibrium value $\bar \alpha=(p^N,0)$ is shown with a red circle, and the joint monopoly optimal belief $(1,0)$ with a red x. From an initial condition very close the equilibrium, the mean dynamics lead on an outward trajectory in the direction of the joint monopoly.  The point along these trajectories which is farthest from the equilibrium is what we identified as $\alpha^e$ above, after which the mean dynamics point back toward the equilibrium $\bar \alpha$. While the beliefs do not transit all the way to the joint monopoly, when the noise is smaller, they approach it more closely.

\end{document}